\documentclass[journal,onecolumn]{IEEEtran}
\usepackage{amsmath} 
\usepackage{amssymb}
\usepackage{booktabs} 
\usepackage{amsthm}
\usepackage{xcolor}
\usepackage{subfigure} 
\usepackage{threeparttable}
\usepackage{tabularx}
\usepackage{cite}
\usepackage{multirow}

\newtheoremstyle{mytheoremstyle}  
  {\topsep}   
  {\topsep} 
  {}   
  {}        
  {\bfseries\upshape} 
  {.}        
  {0.5em}  
  {}      
\theoremstyle{mytheoremstyle}  
\newtheorem{theorem}{Theorem}
\newtheorem{definition}{Definition}
\newtheorem{lemma}{Lemma}

\newtheorem{remark}{Remark}
\usepackage{scalerel}
\newcommand{\leadv}{{{\scaleto{\mathrm{L}}{3.5pt}}}} 
\newcommand{\followv}{{{\scaleto{\mathrm{F}}{3.5pt}}}}
\newcommand{\human}{{{\scaleto{\mathrm{H}}{3.5pt}}}}
\newcommand{\av}{{{\scaleto{\mathrm{AV}}{3.5pt}}}}
\newcommand{\egov}{{{\scaleto{\mathrm{E}}{3.5pt}}}}
\begin{document}

\title{Shared Control for Vehicle Lane-Changing with Uncertain Driver Behaviors}

\author{
    \IEEEauthorblockN{Jiamin Wu, Chenguang Zhao, Huan Yu\textsuperscript{*}} \\
    \IEEEauthorblockA{Thrust of Intelligent Transportation, The Hong Kong University of Science and Technology (Guangzhou)\\ Nansha, Guangzhou, 511400, Guangdong, China
    \thanks{\textsuperscript{*} Corresponding author. Email: huanyu@ust.hk}}
}

\maketitle

\begin{abstract}
Lane changes are common yet challenging driving maneuvers that require continuous decision-making and dynamic interaction with surrounding vehicles. Relying solely on human drivers for lane-changing can lead to traffic disturbances due to the stochastic nature of human behavior and its variability under different task demands. Such uncertainties may significantly degrade traffic string stability, which is critical for suppressing disturbance propagation and ensuring smooth merging of the lane-changing vehicles. This paper presents a human–automation shared lane-changing control framework that preserves driver authority while allowing automated assistance to achieve stable maneuvers in the presence of driver's behavioral uncertainty. Human driving behavior is modeled as a Markov jump process with transitions driven by task difficulty, providing a tractable representation of stochastic state switching. Based on this model, we first design a nominal stabilizing controller that guarantees stochastic $\mathcal{L}_2$ string stability under imperfect mode estimation. To further balance performance and automated effort, we then develop a Minimal Intervention Controller (MIC) that retains acceptable stability while limiting automation. Simulations using lane-changing data from the NGSIM dataset verify that the nominal controller reduces speed perturbations and shorten lane-changing time, while the MIC further reduces automated effort and enhances comfort but with moderate stability and efficiency loss. Validations on the TGSIM dataset with SAE Level 2 vehicles show that the MIC enables earlier lane changes than Level 2 control while preserving driver authority with a slight stability compromise. These findings highlight the potential of shared control strategies to balance stability, efficiency, and driver acceptance.
\end{abstract} 

\begin{IEEEkeywords}
Lane-changing control, shared control, stochastic $\mathcal{L}_2$ string stability, minimal intervention
\end{IEEEkeywords}

\section{Introduction}
Lane changing is a critical traffic maneuver that often serves as a contributory factor in roadway crashes and flow disruptions~\cite{ali2021clacd,MONTEIRO2023104138}. It is a coupled lateral–longitudinal action in which lateral movement across lane boundaries realizes the geometric crossing, while longitudinal adjustment of speed and gap determines whether the maneuver can be executed within the available time window \cite{nilsson2017lane,luo2016dynamic,zhao2024safe,PAN2016403}. A successful execution therefore depends on an acceptable spatiotemporal window, compatibility of speeds and gaps at the crossing instant, and coordination with neighboring vehicles without increasing negative traffic impacts \cite{bevly2016lane,hu2019trajectory}. Human drivers exhibit substantial variability in throttle and brake inputs, which can magnify speed disturbances. In contrast, fully automated control can suppress such fluctuations through reliable state feedback from surrounding vehicles, but it still faces constraints of trust, acceptance, and responsibility in real deployment \cite{fang2023humanmachine,ZHANG2019207,DU2019428}. These realities make human–automation shared control a practical transitional approach in which automation provides suitable assistance without undermining driver authority and engagement. 

Considering the determinant role of longitudinal gap and speed regulation in lane changing and the fact that throttle and brake constitute the most intuitive control channel for drivers, longitudinal shared control becomes a natural choice to both secure maneuver feasibility and preserve driver authority \cite{guo2015sharedmerge}. In this paper, we focus on a typical lane-changing scenario that involves a leader vehicle, a follower vehicle and a lane-changing ego-vehicle. The ego-vehicle must adjust its longitudinal gap under constrained spatiotemporal conditions, as illustrated in Fig.~\ref{fig:scenario}. We develop a shared control framework to ensure maneuver feasibility, preserve driver authority, and attenuate disturbance propagation under stochastic driver behavior.

\subsection{Lane-changing control on automated vehicles}
Automated vehicles (AVs) have already demonstrated the capability to execute automatic lane changes. Commercial vehicles equipped with high-level automation, such as SAE Level 4 and above, can autonomously plan trajectories and complete lane changes within defined operational domains without human intervention \cite{YU2021103101}. Compared to human-driven vehicles (HVs), AVs can perform these maneuvers more smoothly and safely in many traffic scenarios, which improves traffic efficiency \cite{YANG2018228}. To realize such benefits in a systematic manner, many control strategies have been developed, typically organized as a pipeline that perceives surrounding traffic, trajectories generation and tracking by controllers such as linear quadratic regulator (LQR) or model predictive control (MPC) \cite{nilsson2017lane,liu2018dynamic,WANG201573}. In addition, hierarchical approaches integrate high-level behavioral decisions, such as target vehicle or lane selection, with low-level trajectory generation, enabling more efficient and smoother lane-changing maneuvers \cite{Tran2019mpc,li2020automatic}. Despite these advances, single-vehicle formulations often treat other vehicles as moving obstacles, which limits their adaptability in interactive and flow-level traffic.

To overcome these limitations, some studies have developed control strategies that explicitly incorporate surrounding vehicles into the control system \cite{Wang2023Interaction,bhattacharyya2023automated}. For example, an adaptive mixed-integer MPC scheme has been proposed to estimate the cost weights of neighboring vehicles online via inverse optimal control and jointly optimize the trajectories of both the AV and the interacting human-driven vehicle \cite{bhattacharyya2023automated}, with the strategy further validated in controlled experiments \cite{Bhattacharyya2025LaneChange}. Another representative approach~\cite{zhao2024safe} formulates a mandatory lane-changing task within a unified framework that integrates control barrier functions and signal temporal logic, jointly handling safety constraints and timing requirements in the presence of surrounding vehicles, thereby achieving safe and efficient maneuvers. 
Some approaches model the lane-changing problem as a game in which the AV acts as a strategic player that optimizes its acceleration and timing while anticipating the reactions of nearby vehicles, thereby achieving cooperative gap negotiation \cite{YU2018gametheory,wei2022game,Zhang2024Stackelberg}. Meanwhile, deep reinforcement learning (DRL) methods formulate lane changing as a sequential decision-making problem, where the AV learns coupled longitudinal–lateral actions and searches traffic gaps to maximize long-term rewards combining safety, efficiency, and comfort \cite{Triest2020Learning,Saxena2020DrivinginDense,Ye2020RLlanechanging,Li2023DeepReinforcementLearning}, but these data-driven approaches usually lack theoretical guarantees. Overall, these system-level and interactive approaches highlight the potential of fully automated controllers but also face practical limitations, particularly in terms of limited user trust, which relates to user's concerns about loss of agency and unpredictability of automated decisions. Therefore, even though high-level automation assumes full control, users often disengage and take over when uncertainties arise \cite{DU2019428}.

Consequently, the gap between the technical feasibility of automation and its real-world acceptance highlights the need for paradigms that retain human involvement. Motivated by this perspective, shared control frameworks have been proposed, in which automation assists drivers in accomplishing driving tasks while preserving certain level of driver authority. 

\subsection{Shared control methods and applications}
Shared control is defined as a framework in which human and automated control inputs are combined to jointly influence system behavior, enabling the system to leverage both the effectiveness of automation and the adaptability of human~\cite{li2021indirect}. Unlike purely manual or fully automated driving, shared control requires two essential designs: how to model the human's behavior, which is inherently variable and stochastic in most cases; and how to fuse human and automated control commands into a final input \cite{abbink2012haptic}. Existing shared control methods can be broadly categorized into two approaches: switching control and blending control. 
In switching control, the authority is switched between human and automated system based on contexts or thresholds~\cite{LI2020switchingcontrol,HUANG2025105262}. The main drawback is frequent authority transfer, which may degrade both human satisfaction and system performance~\cite{kim2012switchingcons,GAO2024104491}.
In blending control, human and automated inputs are combined in real time~\cite{Dragan2013APF,JIANG2020safeteleoperation,huang2024safetyaware}. Compared with switching control, blending control is more flexible and thus has become the dominant paradigm in shared control applications.

Many research works have applied shared control design for automated systems to ensure safer and more efficient operations during various driving tasks, such as lane-keeping, path-following, and car-following~\cite{XING2021humancenteredcollaborative,erlien2016sharedsteering,tsoi2010lanekeeping,Sentouh2019driverautomation,bian2020lanekeeping}. 
In lane-keeping assistance, the control goal is to maintain a reference trajectory while reducing workload and preventing unsafe departures. One representative design combines lane-keeping and conflict management under a supervisory control framework that guarantees closed-loop stability and reduces conflicts between driver and automation \cite{Sentouh2019driverautomation}. The work in~\cite{bian2020lanekeeping} proposes a learning-based predictive controller with switchable assistance modes, in which an oracle compensates for unmodeled vehicle dynamics while reducing the burden of frequent steering corrections to improve driver acceptance. These methods demonstrate the effectiveness of shared control for enhancing safety in lane keeping. But since lane-keeping is relatively structured with limited degrees of freedom, they mainly emphasize stability and human-automation conflict mitigation rather than explicit driver modeling. 
In path-following tasks, where the reference may involve more complex geometries and uncertainties than simple lane keeping, shared control strategies incorporate explicit descriptions of human behavior. Some approaches adaptively reallocate authority in real time based on workload indicators or control errors, thereby balancing human and automation contributions to enhance trajectory tracking, improve stability, and reduce conflicts \cite{li2022adaptive,shi2022human,huang2024safetyaware}. Other studies emphasize interaction modeling in which the driver and automation are treated as coupled decision-makers, and adaptive adjustment mechanisms are designed to preserve driver authority while enhancing comfort and maintaining safety \cite{han2022adaptivesteering,guo2024game}. In addition, probabilistic intent prediction has been explored, where Gaussian process models of human behavior are fused with automated MPC signals to improve path-tracking and heading performance while reducing risk \cite{lang2023shared}. Overall, these path-following studies highlight the importance of modeling human variability but remain limited to short-term or single-vehicle assistance. To broaden this scope, recent work has incorporated shared control into longitudinal car-following by modeling drivers with the Intelligent Driver Model (IDM) and blending their actions with feedback controllers through a smooth arbitration function. The goal is to suppress stop-and-go oscillations and guarantee string stability. Microscopic simulations further demonstrate that even partial penetration of such shared control strategies can effectively attenuate traffic waves \cite{salvato2024stopandgo}.

Researchers have also investigated shared control in more complex maneuvers such as lane changing, where strict spatiotemporal constraints impose higher cognitive demands on drivers, making shared strategies particularly valuable. In~\cite{yang2018stackelberg}, lane changing is formulated as a Stackelberg game, with the driver as leader and the automation as follower. The automation follows the driver's intended path to support maneuver execution, thereby reducing steering effort and enabling safe overtaking or collision avoidance. 
Moving beyond simple authority negotiation, the human-like shared driving strategy in~\cite{liu2023Shared} explicitly models human perception. It builds risk-assessment models and critical-gap estimators from naturalistic data to allocate authority, and combines them with a cooperative formulation to jointly stabilize lateral and longitudinal maneuvers under time-varying speeds.
Likewise, a collaborative game-theoretic framework models driver and automation across both dimensions under a Nash-equilibrium formulation, dynamically reallocating authority based on risk fields to enhance tracking performance, reduce driver burden, and maintain stability across different speeds \cite{chen2025Collaborative}.
To further formalize lane-changing constraints, \cite{chen2025dynamic} formulates the task as an optimization-based shared control problem, embedding control barrier and Lyapunov functions in a quadratic program to achieve dynamic authority allocation between driver and automation. This design ensures that the automation intervenes only when lane-changing constraints are at risk, while minimizing conflicts and preserving driver authority.

Overall, existing methods realize lane-changing control mainly through steering coordination or authority allocation, typically focusing on lateral control but overlooking longitudinal gap regulation that critically shapes follower-vehicle responses and system-level string stability. In addition, their modeling of human behavior is often simplified, without fully capturing how drivers adapt their actions to dynamic traffic conditions. In this paper, we propose a human-automation shared control framework for lane changing, where human behavior is formulated as a mode-switching stochastic process and the assistive controller is designed to guarantee system-level string stability under uncertain driver inputs.
\begin{figure}
    \centering
    \includegraphics[width=0.8\linewidth]{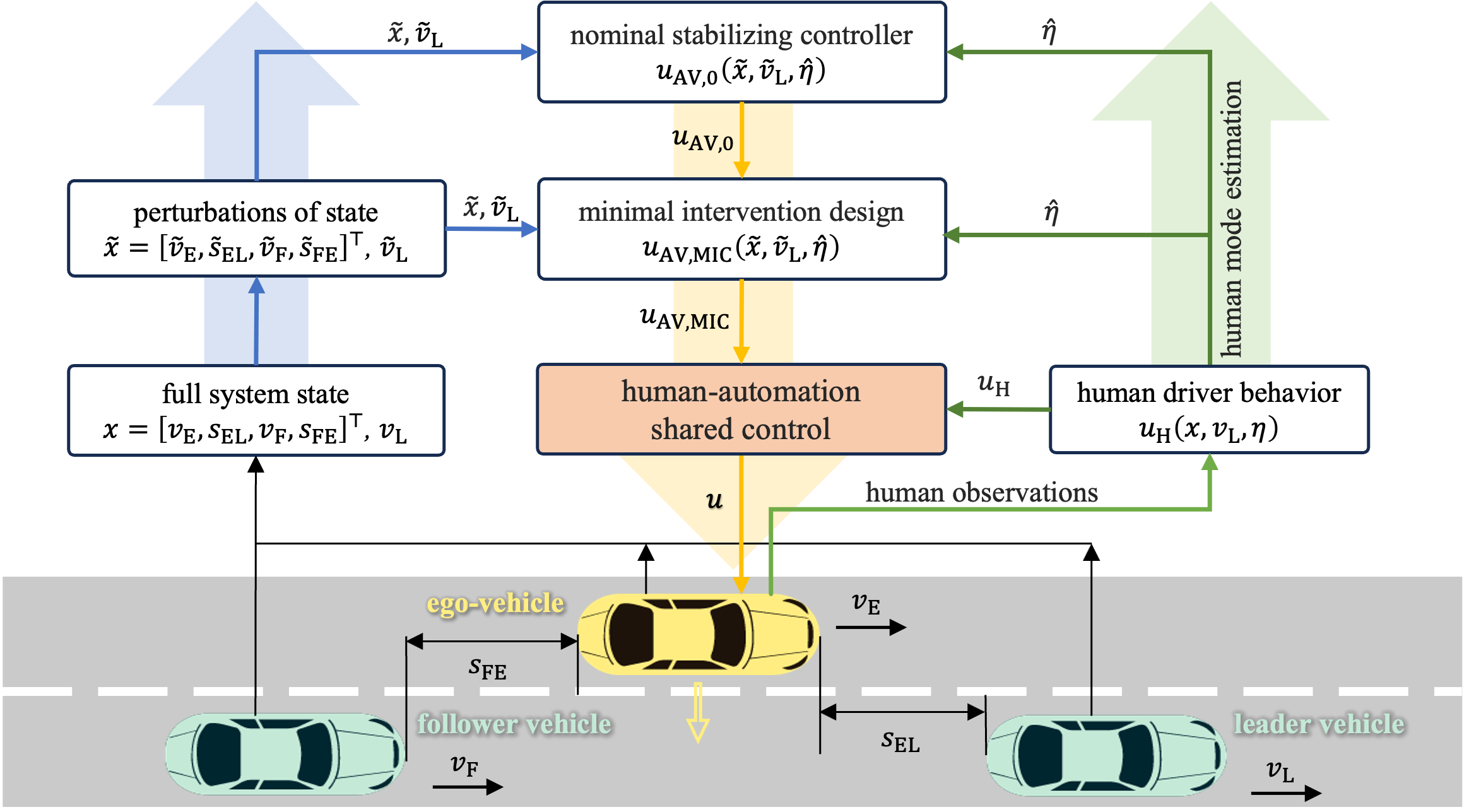}
    \caption{The ego-vehicle adjusts its speed and gap, then changes to the target lane when it has a suitable gap with both the leader vehicle and the follower vehicle. The entire figure illustrates the proposed shared control framework: the human driver generates inputs based on observations, while the assistive controller provides assistance through either the nominal stabilizing controller or the minimal intervention design, both of which adapt to the estimated human mode $\hat{\eta}$. The combined input $u=u_\human+u_\av$ drives the ego-vehicle, ensuring disturbance attenuation and maintaining driver authority under stochastic behavior transitions.}
    \label{fig:scenario}
\end{figure}

\subsection{Contributions}
In this paper, we consider a lane-changing scenario shown in Fig.~\ref{fig:scenario} that contains two HVs (a leader vehicle and a follower vehicle) and an ego-vehicle. By leveraging the information from onboard sensors, we design an assistive controller for the ego-vehicle to stabilize the motion of the three-vehicle system during lane changes, while ensuring the authority of the human driver. We first model human's driving mode as a joint Markov process. Then, we design an assistive controller that adapts to the observed driver mode in real time. We use tractable linear matrix inequalities (LMIs) to guarantee stochastic $\mathcal{L}_2$ string stability under imperfect mode estimation, and explicitly limit automated effort via a minimal intervention principle when possible. Finally, we conduct numerical simulations to validate our assistive controller and analyze its impact on stability, lane-changing success, and driver authority.

The main contributions of this paper lies in proposing a shared control framework for longitudinal dynamics in lane-changing that guarantees stochastic $\mathcal{L}_2$ string stability of the traffic system under driver behavior uncertainty. In particular, we adopt a partially observed Markov representation of human longitudinal actions, which accounts for stochastic switching induced by traffic perception and thus provides a richer description of driver variability compared with deterministic models. We then design an assistive controller that maintains stability under partial observability. Furthermore, we incorporate a minimal intervention mechanism that explicitly penalizes automated effort to preserve driver authority. This yields a mode-dependent controller that constrains automation involvement while ensuring stability. Compared with human-only and automation-only baselines, the proposed design mitigates follower vehicle speed oscillations and shortens maneuver duration while keeping automated intervention bounded.

The remainder of this paper is organized as follows. We formulate the model of the vehicle lane-changing task and human driving behavior in Section~\ref{sec:formulation}. We establish an assistive controller and address stochastic $\mathcal{L}_2$ string stability analysis of the system in Section~\ref{sec:stab}. In Section~\ref{sec:MIC}, we propose a minimal intervention controller designs to guarantee the balance between automated effort and system performance. We conduct numerical simulations in Section~\ref{sec:sim} to validate and analyze the designed controller.

\section{Problem formulation}
\label{sec:formulation}
In this section, we formulate the lane-changing scenario and describe the system dynamics, and then model the human driving behavior.

\subsection{Vehicle motion modeling and system descriptions}
We consider the lane-changing process in a two-lane scenario as illustrated in Fig.~\ref{fig:scenario}, where a leader vehicle and a follower vehicle travel on the same lane and a lane-changing ego-vehicle on the adjacent lane aims to merge into the gap between them. A high-level planner determines the insertion pair and specifies a finite crossing-time window, during which the maneuver must be completed. Since lateral tracking is largely geometric once a gap is available, we focus on longitudinal regulation, which determines whether the maneuver is acceptable and whether disturbances are amplified in the traffic stream. In this work, we model the ego-vehicle's longitudinal input as a blending of human input and automated assistance. This formulation allows explicit quantification of both contributions while regulating speeds and gaps during lane changes.

The ego-vehicle measures the speed and position of the leader vehicle and follower vehicle through onboard or roadside sensing. The state variable of the control system $x(t) \in \mathbb{R}^4$ is defined as:
\begin{equation}
    x = [v_\egov,s_{\egov\leadv},v_\followv,s_{\followv\egov}]^\top,
    \label{eq:state}
\end{equation}
where $v_\egov(t) \in \mathbb{R}$ and $v_\followv(t) \in \mathbb{R}$ denote the speeds of the ego-vehicle and follower vehicle, respectively; $s_{\egov\leadv}(t) \in \mathbb{R}$ is the gap between the ego-vehicle and the leader vehicle, and $s_{\followv\egov}(t) \in \mathbb{R}$ is the gap between the follower vehicle and the ego-vehicle. The leader vehicle's speed $v_\leadv(t) \in \mathbb{R}$ is treated as an external disturbance to the system. The gap between the follower vehicle and leader vehicle, denoted by $s_{\followv\leadv}(t)\in \mathbb{R}$, equals the sum of $s_{\followv\egov}(t)$ and $s_{\egov\leadv}(t)$. The longitudinal dynamics of the follower vehicle are described by:
\begin{align}
    &\dot{s}_{\followv\egov} = v_\egov-v_\followv, \\
    &\dot{s}_{\followv\leadv} = v_\leadv - v_\followv,\\
    &\dot{v}_\followv = f_\followv(x,v_\leadv).\label{eq:f_f}
\end{align}
The function $f_\followv: \mathbb{R}^4\times \mathbb{R} \to \mathbb{R}$ models the follower vehicle's acceleration as a function of the system state $x$ and the leader vehicle's speed $v_\leadv$. We leave $f_\followv$ unspecified and allow it to represent general models.
The control input $u(t) \in \mathbb{R}$ controls ego-vehicle's speed, and the longitudinal dynamics of ego-vehicle are:
\begin{align}
  &\dot{s}_{\egov\leadv}= v_\leadv- v_\egov, \\
  &\dot{v}_\egov = u.   
\end{align}
The overall system dynamics are:
\begin{align}
    \dot{x} = f(x,v_\leadv)+g(x)u, 
\end{align}
\begin{align}
    f(x,v_\leadv) = \begin{bmatrix}
        0 \\ v_\leadv - v_{\egov} \\ f_\followv(x,v_\leadv) \\ v_{\egov} - v_{\followv}
    \end{bmatrix}, \quad  g(x) = \begin{bmatrix}
        1 \\ 0 \\ 0 \\0
    \end{bmatrix}.
\end{align}

We adopt a shared control scheme in which the control input $u(t)$ consists of two parts: a human control input $u_\human(t) \in \mathbb{R}$ and an automated system input $u_\av(t) \in \mathbb{R}$:
\begin{equation}
    u = f_u(u_\human,u_\av),
\end{equation}
where $f_u: \mathbb{R}^2 \to \mathbb{R}$ defines the combination strategy between the human and automated control inputs. The human input $u_\human(t)$ reflects the driver's decision-making process and is modeled by a known function to be introduced later. The assistive input $u_\av(t)$ is designed to assist human decisions and mitigate disturbances in the system. In this work, we adopt an additive combination:
\begin{align}
    u = u_\human+u_\av,
\end{align}
where the human maintains the control authority via $u_\human(t)$, and the automated system $u_\av(t)$ assists to improve overall performance.

We define the system output as 
\begin{align}
    z = \mathcal{H}(x,u),
\end{align}
where the output function $\mathcal{H}$ collects physically meaningful performance variables in the traffic context. In this study, we will first select the follower vehicle's speed as the performance output for stabilization, and later augment it with a weighted automated effort term for minimal intervention design. Such a choice reflects the quantities that determine disturbance attenuation and human authority in lane-changing system. The explicit form of $\mathcal{H}(\cdot)$ will be specified in the subsequent sections.

\subsection{Human behavior modeling}
Deterministic car-following models can approximate human longitudinal actions but fail to capture the variability arising from cognitive factors such as workload or traffic perception \cite{saifuzzaman2015revisiting,li2017estimation}. The variability becomes especially pronounced in high-demand tasks like lane changing, where previous studies have shown that driver behavior is better represented as stochastic switching among distinct modes rather than continuous deterministic dynamics \cite{wang2017drivingstyleanalysis,huainingwu2022stochasticstability}. To this end, we employ a continuous-time hidden Markov model (CT-HMM), which provides a principled framework for describing stochastic transitions among behavior modes over time \cite{LIU2023stabilitysemimarkov}.

In the lane-changing, we categorize human behavior modes into \textit{low task difficulty} and \textit{high task difficulty}, reflecting the instantaneous challenge of the driving task relative to the driver's cognitive capacity. The task difficulty may fluctuate due to both internal psychological factors and reactions to the traffic environment. We formulate these two cases by a mode index $\eta(t) \in \mathcal{N} = \{1,2\}$. To reflect human's underlying cognitive factors, the dynamics of $\eta(t)$ are modeled as a time-homogeneous Markov process governed by transition rates $\Lambda = [\lambda_{ij}]\in\mathbb{R}^{2\times 2}$.

However, since the true human behavior mode $\eta(t)$ cannot be perfectly observed in practice, we treat $\eta(t)$ as a hidden mode and introduce an observed mode $\hat{\eta}(t)$ taking values in $\mathcal{M} = \{1,2\}$. To capture the coupled evolution of the hidden and observed modes, we define a continuous-time hidden Markov model $Z(t) = (\eta(t),\hat{\eta}(t))$, which is a time-homogeneous Markov process defined over the joint space $\mathcal{N} \times \mathcal{M}$ and characterized by transition rates $\nu_{(i,k)(j,l)}$. The transition probabilities follow \cite{stadtmann2017h_2}:
\begin{align}
    &P(Z(t+h) = (j,l)|Z(t) = (i,k))  
    =  \begin{cases}
        \nu_{(i,k)(j,l)}h+o(h),&\quad (j,l) \neq (i,k) \\
        1+\nu_{(i,k)(i,k)}h+o(h),&\quad (j,l) = (i,k).
    \end{cases} 
\end{align}
where $h$ is a small time step, and $o(h)$ is a higher-order term with $o(h)/h \to 0$ as $h\to0$. The transition rates $\nu_{(i,k)(j,l)}$:
\begin{align}
    \nu_{(i,k)(j,l)} = \begin{cases}
        \alpha_{jl}^k \lambda_{ij},& j\neq i, l\in \mathcal{M}, \\
        q_{kl}^i,&j= i, l\neq k,\\
        \lambda_{ii} + q_{kk}^i, &j= i, l=k,
    \end{cases}
\end{align}
where $\sum_{l\in \mathcal{M}} \alpha_{jl}^k = 1$, $\lambda_{ij}\ge0$ for all $i\neq j$, $q_{kl}^i\ge 0$ for all $l \neq k$, and $\lambda_{ii} = -\sum_{j \ne i} \lambda_{ij}$, $q_{kk}^i = -\sum_{l \ne k} q_{kl}^i$. Here, $\alpha_{jl}^k$ captures how likely the controller is to misclassify the new human mode upon a mode switch, while $q_{kl}^i$ models spontaneous estimation updates. 
We define the invariant set of $Z(t)$ as $\mathcal{V} \subseteq \mathcal{N} \times \mathcal{M}$. The joint process of $Z(t)$ enables the automated control system to reason about both the driver's actual behavior and the estimation uncertainty in a stochastic framework, which is crucial for controller design under imperfect observations.

To model the human behavior under different task conditions, we define the human control input as a mode-dependent acceleration function:
\begin{equation}
u_\human = f_\human(x,v_\leadv, \eta),
\end{equation}
where $f_\human: \mathbb{R}^4 \times \mathbb{R} \times \mathcal{N} \to \mathbb{R}$ is specified by the behavior mode $\eta(t)$. For instance, under different task difficulty levels, the driver may adopt distinct sensitivity gains or desired speeds. By assigning separate parameter sets to each mode, this formulation captures the variability in human behavior across low and high task difficulty states.

\section{Stabilization of assistive controller}\label{sec:stab}
In Section~\ref{sec:formulation}, we have formulated the control system and modeled human control input as $u_\human$. In this section, we design a nominal stabilizing controller $u_\av$ that provides assistance to human driver.

\subsection{Linearization}
We consider small perturbations around the equilibrium states. 
We denote the equilibrium speed as $v^*$. 
The equilibrium gap between the follower vehicle and ego-vehicle $s_{\followv\egov}^*$ and equilibrium gap between the ego-vehicle and leader vehicle $s_{\egov\leadv}^*$ are decided as $f_{\followv}(v^*,s_{\egov\leadv}^*,v^*,s_{\followv\egov}^*,v^*)=0$, and $f_{\human}(v^*,s_{\egov\leadv}^*,v^*,s_{\followv\egov}^*,v^*,\eta)=0$. We assume that the equilibrium gaps are independent of the human behavior mode, and depend only on the equilibrium speed. 

We denote small perturbations around the equilibrium states as: 
\begin{align}
   \tilde{x} &= \begin{bmatrix}
       \tilde{v}_\egov  & \tilde{s}_{\egov\leadv}  & \tilde{v}_\followv  & \tilde{s}_{\followv\egov}    
   \end{bmatrix}^{\top} \nonumber\\
   &= \begin{bmatrix}
       v_\egov - v^* & s_{\egov\leadv} - s_{\egov\leadv}^* & v_\followv - v^* & s_{\followv\egov} - s_{\followv\egov}^* 
   \end{bmatrix}^{\top}.
\end{align}
The human control input $u_\human = f_\human(x, v_\leadv, \eta)$ is linearized around $(x^*,v^*)$ and is dependent on $\eta$. Thus, $u_\human(\eta)$ is written in a compact form as:
\begin{align}
    u_{\human}(\eta) = K_{\human}(\eta)\tilde{x} + D_{\human}(\eta) \tilde{v}_\leadv,
    \label{eq:uh}
\end{align}
with 
\begin{align}
    K_{\human}(\eta) &= \begin{bmatrix}
        \frac{\partial f_{\human}(\eta)}{\partial v_{\egov}} & \frac{\partial f_{\human}(\eta)}{\partial s_{\egov\leadv}} &  \frac{\partial f_{\human}(\eta)}{\partial v_{\followv}} & \frac{\partial f_{\human}(\eta)}{\partial s_{\followv\egov}} \end{bmatrix}, \\
    D_{\human}(\eta) &= \frac{\partial f_{\human}(\eta)}{\partial v_\leadv}, 
\end{align}
where $\tilde v_{\leadv} = v_{\leadv} - v^*$. Likewise, the dynamics of the follower vehicle $f_{\followv}$ are linearized around $(x^*,v^*)$. Substituting $u = u_\human + u_\av$ into the system yields the following linearized dynamics:
\begin{align}
    \dot{\tilde{x}} &= [A+ BK_{\human}(\eta)]\tilde{x}  +Bu_{\av} + [D + BD_{\human}(\eta)] \tilde{v}_\leadv, 
    \label{eq:system}
\end{align}
with
\begin{align}
    A = \begin{bmatrix}
        0 & 0 &0 &0  \\
        -1 & 0 & 0 & 0   \\
        \frac{\partial f_{\followv}}{\partial v_{\egov}}  & \frac{\partial f_{\followv}}{\partial s_{\egov\leadv}}  & \frac{\partial f_{\followv}}{\partial v_{\followv}} & \frac{\partial f_{\followv}}{\partial s_{\followv\egov}}    \\
        1 & 0 & -1 & 0 
    \end{bmatrix}, 
    B = \begin{bmatrix}
        1  \\
        0  \\
        0 \\
        0  
    \end{bmatrix}, D = \begin{bmatrix}
        0 \\ 1 \\ \frac{\partial f_{\followv}}{\partial v_\leadv} \\0 
    \end{bmatrix}.
\end{align}
We note that  $A$, $B$, and $D$ are constant system matrices independent of the behavior mode $\eta$. We assume that $(A, B)$ is stabilizable and $(A,C)$ is detectable so that no unstable dynamics are unobservable from the performance output.

The control goal is to stabilize the follower vehicle's motion, so we define the system output as the follower vehicle's speed perturbation:
\begin{align}
    \tilde{z} = C \tilde{x},
\end{align}
where $C = \begin{bmatrix} 0 & 0 & 1 & 0 \end{bmatrix}$.

Building on the joint modeling of $\eta(t)$ and $\hat{\eta}(t)$, we design the assistive controller based on the observed mode $\hat{\eta}(t)$, while accounting for possible mismatch with the true human mode. The AV assistance control law is:
\begin{align}
     u_{\av}(\hat{\eta}) &= K_{\av}(\hat{\eta})\tilde{x} + D_{\av}(\hat{\eta}) \tilde{v}_\leadv \label{eq:ua},
\end{align}
where $K_{\av}(\hat{\eta})$ and $D_{\av}(\hat{\eta})$ are controller gains to be designed.

To summarize, from equations \eqref{eq:uh}, \eqref{eq:system} and \eqref{eq:ua}, the closed-loop dynamics are:
\begin{align}
    \dot{\tilde{x}} &= [A+ BK_{\human}(\eta)+BK_{\av}(\hat{\eta})]\tilde{x}  + [D + BD_{\human}(\eta) + BD_{\av}(\hat{\eta})] \tilde{v}_\leadv.
    \label{eq:system2}
\end{align}
Since the human driver's behavior mode $\eta$ is formulated by a discrete finite space $\mathcal{N}$, for simplified notation, we denote $K_{\human}(\eta)=K_{\human,i} =\begin{bmatrix}\frac{\partial f_{\human,i}}{\partial v_{\egov}} & \frac{\partial f_{\human,i}}{\partial s_{\egov\leadv}} &  \frac{\partial f_{\human,i}}{\partial v_{\followv}} & \frac{\partial f_{\human,i}}{\partial s_{\followv\egov}} \end{bmatrix}$, $D_{\human}(\eta)=D_{\human,i} = \frac{\partial f_{\human,i}}{\partial v_\leadv}$, for $\eta=i \in \mathcal{N}$. Similarly, we have $K_{\av}(\hat{\eta})=K_{\av,k}$ and $D_{\av}(\hat{\eta})=D_{\av,k}$, when $\hat{\eta}=k \in \mathcal{M}$. The linearized control system~\eqref{eq:system2} is therefore a continuous-time Markov jump system with imperfect observations.

\subsection{Stochastic $\mathcal{L}_2$ string stability}
In the lane-changing scenario, the leader vehicle acts as an external disturbance and fluctuations in its velocity will affect the behaviors of following ego-vehicle and follower vehicle. Meanwhile, the human driver of ego-vehicle may switch behavior across different driving modes. This introduces stochasticity into the system dynamics. As a result, classical string stability analysis becomes insufficient to characterize the propagation of speed perturbations under stochasticity. To guide the design of the assistive controller under such uncertainty, we adopt the stochastic $\mathcal{L}_2$ string stability, which evaluates whether the propagation of upstream disturbances can be attenuated in expectation. 
\begin{definition}[Stochastic $\mathcal{L}_2$ string stability \cite{Li2019stringstability}]\label{df:df1}
Consider the lane-changing system in \eqref{eq:system2}, the system is stochastic $\mathcal{L}_2$ string stable, 
if there exist class $\mathcal{K}$ functions $\varphi$ and $\psi$ such that, for any initial state $\tilde{x}(0) \in \mathbb{R}^4$ and any $\tilde{v}_\leadv \in \mathcal{L}_2$:
\begin{align}
    \mathbb{E} \left\{ \| \tilde{z}\| _{\mathcal{L}_2} \right\} \leq \varphi \left(\mathbb{E} \left\{ \| \tilde{v}_\leadv \| _{\mathcal{L}_2} \right\} \right) + \psi \left(\|\tilde{x}(0)\|\right),
\end{align}
where the expectation $\mathbb{E}$ is taken with respect to the joint Markov process $Z(t)$.
\end{definition}

\begin{lemma}\cite{Li2019stringstability} \label{lem:L2}
     For the lane-changing system~\eqref{eq:system2}, when the initial condition is $\tilde{x}(0)=0$, it is stochastically $\mathcal{L}_2$ string stable, if and only if there exists $\gamma\le 1$ such that:
\begin{align}
\mathbb{E} \left\{ \int_{0}^{\infty} \tilde{z}^\top(s)  \tilde{z}(s) ds \right\} \leq \gamma^2 \mathbb{E} \left\{ \int_{0}^{\infty} \tilde{v}_\leadv^\top(s)  \tilde{v}_\leadv(s) ds \right\}.
\end{align}
\end{lemma}

For the system with only human control, Lemma \ref{lem:L2} may not hold, since unstable human behavior can amplify disturbances. Therefore, it is necessary to design an assistive controller to achieve stability criterion in Lemma \ref{lem:L2}. The following theorem establishes a set of linear matrix inequalities (LMIs) that guarantees stochastic $\mathcal{L}_2$ string stability through appropriate controller gain design.

\begin{theorem}\label{thm:thm1}
Consider the closed-loop system~\eqref{eq:system2} under the joint process $Z(t) = (\eta(t), \hat{\eta}(t)) \in \mathcal{N} \times \mathcal{M}$ with joint transition rates $\nu_{(i,k)(j,l)}$. For each $k \in \mathcal{M}$, the control gains:
\begin{align}
    K_{\av,k}= V_k G_{k}^{-1}, \quad D_{\av,k} = L_k,
\end{align}
render the system stochastically $\mathcal{L}_2$ string stable, where there exist matrices $X_{ik}>0$, $V_k$, $L_k$, and non-singular $G_k$ for $i\in\mathcal{N}$ and $k\in\mathcal{M}$, and a scalar $\varepsilon>0$, satisfying the following LMIs\footnote{For a square matrix $A$, the operator $\mathrm{Her}(A) = A + A^\top$.}:
\begin{align}\label{eq:LMI0}
    \begin{bmatrix}
        \nu_{(i,k)(i,k)}X_{ik} & \Phi_{ik} & 0 &X_{ik} & X_{ik}\Pi_{ik} \\
        \Phi_{ik}^\top & -\gamma^2I & 0 &0 & 0\\
        0 & 0 & -I & 0 & 0\\
        X_{ik}^\top & 0 &0 & 0 & 0\\
        \Pi_{ik}^\top X_{ik}^\top &0 & 0 & 0 & -\mathcal{D}_{ik}
    \end{bmatrix} + \mathrm{Her}\left(\begin{bmatrix}
         \Omega_{ik}\\
        0 \\
        CG_k \\
        -G_k \\
        0
    \end{bmatrix}\begin{bmatrix}
        \varepsilon I\\
        0\\
        0\\
        I\\
        0
    \end{bmatrix}^\top \right)<0
\end{align}
where
\begin{align}
    \Omega_{ik} = &(A+BK_{\human,i})G_k +BV_k \label{eq:LMI_con1}\\
    \Phi_{ik} =& D+BD_{\human,i}+BL_k \label{eq:LMI_con2}\\
    \mathcal{V}_{(i,k)} =& \{ (j,l)\in \mathcal{V}; (j,l) \ne (i,k)\text{ and } \nu_{(i,k)(j,l)} \ne 0\} \nonumber \\
     =& \{r_{(i,k)}(1),\ldots,r_{(i,k)}(\tau_{(i,k)});r_{(i,k)}(\iota)\in \mathcal{V},\iota  = 1,\ldots, \tau_{(i,k)}\} \label{eq:LMI_con3}\\
    \Pi_{ik} =& \begin{bmatrix}
        \sqrt{\nu_{(i,k)r_{(i,k)}(1)}}I & \ldots & \sqrt{\nu_{(i,k)r_{(i,k)}(\tau_{(i,k)})}}I
    \end{bmatrix} \label{eq:LMI_con4}\\
    \mathcal{D}_{ik}=& \mathrm{diag}(X_{r_{(i,k)}(1)},\ldots,X_{r_{(i,k)}(\tau_{(i,k)})}) \label{eq:LMI_con5}
\end{align}
\end{theorem}

\begin{proof}
    See Appendix~\ref{app:proof1}.
\end{proof}

\begin{remark}
The scalar $\gamma$ characterizes the attenuation level of upstream disturbances. A smaller value of $\gamma$ corresponds to a tighter performance requirement, meaning the assistive controller is expected with stronger suppression effect of the leader vehicle's velocity fluctuation. In practice, tuning $\gamma$ balances performance against the feasibility of the LMI solution.
\end{remark}

While the controller in Theorem~\ref{thm:thm1} renders the system string stable, it may intervene excessively and compromise driver authority. To address this limitation, we treat the controller synthesized here as a nominal baseline and in the next section develop a minimal intervention controller that explicitly balances disturbance attenuation with limited automated effort. 

\section{Minimal Intervention Controller Design}\label{sec:MIC}
In Section~\ref{sec:stab}, we establish a nominal assistive controller that guarantees stochastic $\mathcal{L}_2$ string stability of the lane-changing system. In this section, we extend the baseline by developing a minimal intervention controller (MIC), which builds on the nominal controller but introduces an explicit penalty on automated effort to achieve stability with minimal automated intervention.

\subsection{Balance between stability performance and automated effort}
While the nominal controller in Theorem~\ref{thm:thm1} provides the necessary stability guarantees, it does not regulate the magnitude of automated effort, which is a critical aspect of shared control. Excessive or persistent intervention may compromise driver authority, whereas insufficient intervention may fail to stabilize the system. To address this limitation, we augment the performance output with a penalty on the assistive control input, so that the MIC explicitly balances disturbance attenuation with minimal automated intervention. The augmented performance output is defined as
\begin{align}
    \tilde{z}_{\text{aug}} = \begin{bmatrix}
        \tilde{z} \\
        \tilde{z}_{\text{eff}}
    \end{bmatrix} =\begin{bmatrix}
        C\tilde{x} \\
        \beta u_\av
    \end{bmatrix}
    = \begin{bmatrix}
        C\tilde{x} \\
        \beta (K_{\av,k}\tilde{x}+D_{\av,k}\tilde{v}_\leadv)
    \end{bmatrix} = \begin{bmatrix}
        C  \\
        \beta K_{\av,k} 
    \end{bmatrix}\tilde{x}  + \begin{bmatrix}
        0\\
        \beta D_{\av,k}
    \end{bmatrix}\tilde{v}_\leadv. \label{eq:aug_output}
\end{align}
Here, $\tilde{z}=C\tilde{x}$ retains the original output definition in Section~\ref{sec:stab}, which selects the follower vehicle's velocity perturbation to characterize disturbance propagation. The newly introduced term $\tilde{z}_{\text{eff}}=\beta u_\av$ penalizes automated effort, thereby discouraging unnecessary throttle or braking actions. The weight $\beta>0$ is specified offline as a design parameter. In practice, $\beta$ can be chosen based on traffic conditions or driver characteristics that can potentially be informed from historical driving data. For instance, higher $\beta$ values are appropriate for cautious drivers or dense traffic, where automation should intervene conservatively to preserve comfort and driver authority; whereas smaller $\beta$ values suit aggressive drivers or free-flow conditions that permit stronger automated intervention.

\begin{lemma}[Stochastic $\mathcal{L}_2$ performance under augmented output]\label{prop:mic_perf}
Consider the traffic system with the augmented output \eqref{eq:aug_output}, the system is stochastically $\mathcal{L}_2$ string stable if there exists $\gamma>0$ such that, for any $\tilde{v}_\leadv \in \mathcal{L}_2$ and initial condition $\tilde{x}(0)$,
\begin{align}
 \mathbb{E}\left\{ \int_0^\infty \tilde{v}_{\followv}^2(s)\,ds
+ \beta^2 \int_0^\infty u_\av^2(s)\,ds \right\}
\;\leq\;\gamma^2\,
\mathbb{E}\left\{\int_{0}^{\infty}\tilde{v}_\leadv^\top(s)\tilde{v}_\leadv(s)\,ds\right\}.
\label{eq:H_infity_energy_mic}
\end{align}
\end{lemma}

Compared with Lemma~\ref{lem:L2}, Lemma~\ref{prop:mic_perf} generalizes the performance criterion by replacing the output channel with the augmented form in \eqref{eq:aug_output}. The parameter $\gamma$ thus quantifies a joint bound that reflects both disturbance propagation in the follower vehicle and the penalized assistive control input. The design weight $\beta$ regulates this balance: A larger $\beta$ imposes a stronger penalty on automated effort, effectively restricting automation authority and making the condition more conservative with weaker stability performance, whereas a smaller $\beta$ emphasizes disturbance attenuation, enabling stronger controller corrections and yielding tighter stability bounds.

\subsection{Controller synthesis with minimal intervention}
Based on the augmented output and the performance condition, we next derive sufficient LMI conditions that guarantee stochastic $\mathcal{L}_2$ string stability while minimizing automated intervention.
\begin{theorem}[Stochastic $\mathcal{L}_2$ string stability with minimal intervention]\label{thm:mic_lmi}
    Consider the closed-loop system~\eqref{eq:system2} under the joint process $Z(t) = (\eta(t), \hat{\eta}(t)) \in \mathcal{N} \times \mathcal{M}$ with transition rates $\nu_{(i,k)(j,l)}$, and the augmented output definition in \eqref{eq:aug_output}. For each observed mode $k \in \mathcal{M}$, the control gains:
    \begin{align}
    K_{\av,k} = V_k G_k^{-1},\quad D_{\av,k} = L_k
\end{align}
render the system stochastically $\mathcal{L}_2$ string stable under the augmented output, provided that there exist matrices $X_{ik}>0$, $V_k$, $L_k$, and non-singular $G_k$ for $i\in \mathcal{N}$ and $k\in \mathcal{M}$, and a scalar $\varepsilon>0$, satisfying the following LMIs:
\begin{align}
        \begin{bmatrix}
        \nu_{(i,k)(i,k)} X_{ik} & \Phi_{ik} & 0  & 0 &X_{ik} & X_{ik}\Pi_{ik} \\
        \Phi_{ik}^\top & -\gamma^2 I & 0 & \beta L_k^\top &0 & 0\\
        0 & 0 & -I & 0 & 0 & 0\\
        0 & \beta L_k & 0 & -I & 0 & 0\\
        X_{ik}^\top & 0 &0 & 0 & 0 & 0\\
        \Pi_{ik}^\top X_{ik}^\top &0 & 0 & 0 & 0 & -\mathcal{D}_{ik}
    \end{bmatrix} + \mathrm{Her}\left(\begin{bmatrix} \Omega_{ik}\\
        0 \\
        C G_k \\
       \beta V_k \\
        -G_k \\
        0
    \end{bmatrix} \begin{bmatrix}
        \varepsilon I\\
        0\\
        0\\
        0\\
        I\\
        0
    \end{bmatrix}^\top \right)<0,
    \label{eq:LMI1}
\end{align}
where $\Omega_{ik}$, $\Phi_{ik}$, $\Pi_{ik}$, and $\mathcal{D}_{ik}$ are defined consistently with \eqref{eq:LMI_con1} to \eqref{eq:LMI_con5}.
\end{theorem}

\begin{proof}
    See Appendix~\ref{app:proof2}.
\end{proof}

Theorem~\ref{thm:mic_lmi} therefore provides a tractable synthesis condition for the minimal intervention controller. Beyond ensuring stochastic $\mathcal{L}_2$ string stability, it embeds explicit constraints on automated effort, aligning the mathematical design with the broader principle of preserving driver authority in shared control. In particular, the tunability of $\beta$ enables the controller to span a continuum between minimal automated involvement and more assertive stabilization. This flexibility is essential in traffic scenarios such as lane changes, where both stability and driver acceptance must be maintained. In the synthesis, $\beta$ appears only in the automated-effort channel and enters quadratically after the Schur complement, so larger $\beta$ tightens the feasibility region and increases the theoretical attained bound of $\gamma$. The next section will present simulation results that validate the proposed controller and illustrate the balance between stability and automated effort.

\section{Numerical results}\label{sec:sim}
In this section, we conduct simulations to validate the proposed assistive controller in lane-changing scenarios. We compare simulation results by human-only control, nominal shared control, and MIC to highlight disturbance suppression and stability maintenance. These experiments demonstrate both the baseline performance and the impact of the minimal intervention design.

\subsection{Simulation setup}\label{sec:setup}
For the human behavior mode $\eta(t)$ in Section~\ref{sec:formulation}, we adopt a task difficulty (TD) indicator to quantify the driver's instantaneous cognitive demand and determine the mode. Following \cite{saifuzzaman2015revisiting}, the TD is formulated from measurable quantities such as vehicle speed and spacing to the leader vehicle:
\begin{align}
\text{TD}(t) = \left(\frac{v_\egov(t) \cdot T_{\text{des}}}{(1 - \delta) \cdot s_{\egov\leadv}(t)}\right)^\zeta,
\end{align}
where $T_{\text{des}}$ is the desired time headway, $\delta$ is a risk parameter reflecting tolerance for time pressure and small gaps, and $\zeta$ reflects sensitivity. The TD metric thus captures perceived urgency during lane changes. Typically, task difficulty increases when the available lane-changing gap becomes small or when the ego-vehicle travels at a higher speed, both of which impose greater cognitive demand on the human driver. Conversely, larger gaps and lower speeds lead to a less demanding situation and thus a lower task difficulty value. This relationship allows us to associate observable traffic conditions with variations in human lane-changing behavior. For instance, a high TD value often corresponds to the active execution phase of the lane change, where the driver accelerates and reduces spacing to complete the maneuver within a limited gap. In contrast, when TD is low, the driver is more likely to be cruising steadily or making minor adjustments to restore spacing after completing the lane change. The human behavior mode $\eta(t)$ is determined by thresholding the TD value:
\begin{align}
\eta(t) =
\begin{cases}
1, & \text{if } \text{TD}(t) < \text{TD}_\text{threshold}, \\
2, & \text{otherwise},
\end{cases}
\end{align}
with $\text{TD}_\text{threshold} = 1$. Mode $\eta(t) = 1$ corresponds to low task difficulty, and $\eta(t) = 2$ corresponds to high task difficulty. 

To model human behavior transitions, we use real-world lane-changing trajectories from the Next Generation Simulation (NGSIM) dataset to calibrate the transition rates. Specifically, we use data collected between 5:00 p.m. and 5:15 p.m. on a 500-meter segment of Interstate 80 in Emeryville, California. We calibrate the desired time headway $T_{\text{des}}$ and risk sensitivity parameter $\delta$ to reflect realistic responses under varying traffic conditions. Then, we compute task difficulty profiles and assign behavior modes $\eta(t)$ from the collected trajectories, and the inferred mode sequences are used to estimate the transition rate matrix $\Lambda = [\lambda_{ij}]$ via Maximum Likelihood Estimation.

For the observation mode $\hat{\eta}(t)$, we adopt a symmetric specification for tractability and consistency with our implementation. All $\alpha_{jl}^k$ share the same $2\times2$ structure for each $k\in\{1,2\}$, and all $q_{kl}^i$ share the same generator structure for each $i\in\{1,2\}$. Specifically, for each $k\in\{1,2\}$:
\begin{align}
    \big[\alpha_{jl}^k\big]= \begin{bmatrix}
        1-\alpha & \alpha \\
        \alpha & 1-\alpha
    \end{bmatrix}, \quad 0 \le \alpha \le 1.
\end{align}
which ensures that each row sums to one, thereby defining a valid conditional probability distribution over $\hat{\eta}$ given $\eta$. Likewise, for each $i\in \{1,2\}$,
\begin{align}
    \big[q_{kl}^i\big]=\begin{bmatrix}
-q & q\\
q & -q
\end{bmatrix},
\quad q>0,
\end{align}
which guarantees that each row sums to zero with nonnegative off-diagonal entries, and thus constitutes a valid continuous-time Markov generator over $\hat{\eta}$. We set $\alpha=0.05$ and $q=0.02~\text{s}^{-1}$. In Appendix~\ref{app:sensitivity}, we further conduct sensitivity analysis, which shows that the stochastic $\mathcal{L}_2$ string stability is also guaranteed with other $\alpha$ and $q$ values.

For the human control input, we set it as a mode-dependent Optimal Velocity Model (OVM):
\begin{align}
    u_{\human,i} = a_i(V(s_{\egov\leadv})-v_\egov) + b_i(v_\leadv-v_\egov),
\end{align}
where $a_i$ and $b_i$ denote the sensitivities to the desired velocity difference and relative velocity, respectively. The desired speed function $V(s)$ is formulated as a piecewise-continuous function:
\begin{align}
    V(s) = \begin{cases}
        0,& s\leq s_{\mathrm{st}}\\
        \frac{v_{\max}}{2}\big(1-\cos\big(\pi \frac{s-s_{\mathrm{st}}}{s_{\mathrm{go}} - s_{\mathrm{st}}}\big)\big),& s_{\mathrm{st}}<s<s_{\mathrm{go}}\\
        v_{\max}, &s\ge s_{\mathrm{go}},
    \end{cases}
\end{align}
where $v_{\max}$ is the maximum desired speed, $s_{\mathrm{st}}$ the stopping distance, and $s_{\mathrm{go}}$ the free-flow distance. We retain identical desired velocity profiles across both modes by fixing $v_{\max}$, $s_{\mathrm{st}}$, and $s_{\mathrm{go}}$. Mode-dependent differences are introduced solely through the sensitivity parameters $a_i$ and $b_i$, which determine the driver's responsiveness to desired velocity and relative speed, respectively. In the low task difficulty mode ($\eta = 1$), the driver exhibits increased responsiveness to the desired velocity, as indicated by a higher $a_1$ value. Conversely, in the high task difficulty mode ($\eta = 2$), the driver becomes more sensitive to the relative speed to the leader, as reflected by the larger $b_2$ value, which captures the heightened reactivity under elevated cognitive demand such as during active lane-changing. To accurately reflect human control responses under different task difficulty levels, we calibrate the OVM parameters $(a_1, b_1, a_2, b_2, v_{\max}, s_{\mathrm{st}}, s_{\mathrm{go}})$ using lane-changing trajectories from the same NGSIM dataset. We use the Global Least-Squared Errors Calibration approach as in~\cite{treiber2013microscopic}, where simulated trajectories are forward-integrated and fitted to minimize deviation from empirical observations. The calibrated parameters are summarized in Table~\ref{tab:ngsim_all}.

\begin{table}[t]
\centering
\begin{threeparttable}
\caption{Calibrated driver parameters from the NGSIM dataset}
\label{tab:ngsim_all}
\begin{tabularx}{0.95\linewidth}{@{\hspace{6pt}} l @{\extracolsep{\fill}} c c c l @{\hspace{6pt}}}
\toprule
Category & Symbol & Value & Unit & Description \\
\midrule
TD parameters & $T_{\text{des}}$ & 1.09 & s   & Desired headway \\
    & $\delta$         & 0.30 & –   & Risk sensitivity \\
    & $\zeta$         & 1.00 & –   & Exponent for TD scaling \\
\midrule
Markov transition rates & $\lambda_{11}$ & -0.0454 & s$^{-1}$ & Leaving rate from low TD \\
                        & $\lambda_{12}$ & 0.0454  & s$^{-1}$ & Switching $1 \to 2$ \\
                        & $\lambda_{21}$ & 0.1117  & s$^{-1}$ & Switching $2 \to 1$ \\
                        & $\lambda_{22}$ & -0.1117 & s$^{-1}$ & Leaving rate from high TD \\
\midrule
Mode-dependent OVM & $a_1$ & 0.25  & s$^{-1}$ & Sensitivity to desired speed ($\eta=1$) \\
    & $b_1$ & 0.10  & s$^{-1}$ & Sensitivity to relative speed ($\eta=1$) \\
    & $a_2$ & 0.18  & s$^{-1}$ & Sensitivity to desired speed ($\eta=2$) \\
    & $b_2$ & 0.17  & s$^{-1}$ & Sensitivity to relative speed ($\eta=2$) \\
    & $s_{\mathrm{st}}$ & 3.50 & m   & Stopping gap ($\eta=1,2$) \\
                   & $s_{\mathrm{go}}$ & 20.50 & m   & Free-flow gap ($\eta=1,2$) \\
                   & $v_{\max}$ & 20.00 & m/s & Maximum desired speed ($\eta=1,2$) \\
\bottomrule
\end{tabularx}
\begin{tablenotes}
\footnotesize
\item TD and Markov parameters are calibrated using representative lane-changing trajectories from NGSIM.  
\item OVM parameters are estimated via forward-simulation fitting, showing that when $\eta=1$, drivers are more sensitive to desired speed, while under $\eta=2$, drivers are more sensitive to relative speed.
\end{tablenotes}
\end{threeparttable}
\end{table}

The follower vehicle dynamics $f_{\followv}$ are modeled as a car-following model that dynamically selects its leader based on the relative positions of the ego-vehicle and the leader vehicle. At each time instant, the follower vehicle follows whichever vehicle is in front; if the ego-vehicle is behind, the follower vehicle follows the leader vehicle. The acceleration of the follower vehicle is:
\begin{align}
\dot{v}_{\followv} = \begin{cases}
    a(V(s_{\followv\leadv}) - v_{\followv}) + b(v_{\leadv}-v_{\followv}), \text{ if } 
    s_{\followv\egov}\le0\\
    a(V(s_{\followv\egov}) - v_{\followv}) + b(v_{\egov}-v_{\followv}), \text{ if } s_{\followv\egov}>0
\end{cases}
\label{eq:fv_model}
\end{align}
The equilibrium of the lane-changing system is defined with respect to the case where the follower vehicle follows the ego-vehicle, ensuring a consistent model definition once the ego-vehicle has changed the lane. We calibrate follower vehicle model parameters from NGSIM trajectories as: $a=0.26$ s$^{-1}$, $b=0.09$ s$^{-1}$, $s_{\mathrm{st}}=3.0$ m, $s_{\mathrm{go}}=22.0$ m, $v_{\max}=28.0$ m/s. 

We define the completion criterion for lane-changing based on both the rear gap $s_{\followv\egov}$ and the front gap $s_{\egov\leadv}$. 
Specifically, the ego-vehicle is considered to have fully changed into the target lane at the earliest time $t$ such that
\begin{align}
    s_{\followv\egov}(t) > s_{\text{rear,thr}}, \quad 
    s_{\egov\leadv}(t) > s_{\text{front,thr}}, \quad 
    s_{\egov\leadv} - \tau \big(v_\egov - v_\leadv \big) > 0, \quad s_{\followv\egov} - \tau \big(v_\followv - v_\egov \big)>0,
\end{align}
where $\tau$~s is a safe time-to-collision, $s_{\text{rear,thr}}$ and $s_{\text{front,thr}}$ are the minimum gap for the follower–ego and ego-leader respectively. We adopt thresholds extracted from NGSIM lane-changing trajectories, with $\tau = 1$ s, $s_{\text{rear,thr}}=8.8$ m, and $s_{\text{front,thr}}=7.3$ m.

\subsection{Stability improvement by nominal controller}
We consider a scenario in which all vehicles start from the equilibrium velocity $v^*$ and corresponding gaps $s_{\followv\egov}^*$ and $s_{\egov\leadv}^*$. We set the leader vehicle to have a braking-acceleration maneuver:
\begin{align}
    \dot{v}_\leadv =  
    \begin{cases}
        -a_\leadv & \quad \text{if}\quad t\in [0,t_\leadv],\\
        a_\leadv  & \quad \text{if}\quad t\in (t_\leadv, 2t_\leadv],\\
       0 & \quad \text{otherwise},
    \end{cases}
\end{align}
where $a_\leadv = 2~\text{m/s}^2$, $t_\leadv = 2$ s. This speed profile provides a disturbance for evaluating the flow propagation of the traffic system.

We compare the following two control schemes:
\begin{enumerate}
    \item \textbf{Human-only control}, in which the ego-vehicle is governed solely by the mode-dependent human input $u_{\human}(t)$ under $\eta(t)$.
    \item \textbf{Nominal shared control}, in which the ego-vehicle is co-driven by $u(t)=u_{\human}(t)+u_{\av}(t)$, where the automated action $u_{\av}(t)$ is generated according to the observed mode $\hat{\eta}(t)$ using feedback gains $(K_{\av,k},D_{\av,k})$ obtained from solving a set of LMIs:
    \begin{align}
\label{eq:gamma0_nominal}
&\gamma_{0}= \min_{P_{ik}> 0,\,K_{\av,k},\,D_{\av,k}} \ \gamma \\
     \text{s.t.} \quad &   
\mathcal{LMI}_0\!\big(P_{ik},K_{\av,k},D_{\av,k} \big)\ < 0, \notag
\end{align}
where $\mathcal{LMI}_0(\cdot)< 0$ collects the feasibility conditions stated in Theorem~\ref{thm:thm1}.
\end{enumerate}

The performance on string stability is assessed by the empirical gain as the ratio between the $\mathcal{L}_2$ energy of the follower vehicle's speed perturbation $\tilde v_{\followv}$ and that of the disturbance $\tilde{v}_\leadv$:
\begin{align}
\gamma_{\mathrm{est}}= \frac{\|\,\tilde v_{\followv}\,\|_{\mathcal{L}_2(0,T)}}{\|\,\tilde v_{\leadv}\,\|_{\mathcal{L}_2(0,T)}}.
\end{align}
where $T$ is the simulation horizon. The stability threshold is $\gamma_{\mathrm{est}}=1.0$: values below one imply stochastic $\mathcal{L}_2$ string stability, while values above one imply instability. 

We conduct 100 Monte Carlo simulations with randomized $\eta(t)$ switching paths and set the simulation duration as $T = 20$ s. Table~\ref{tab:L2norm} gives the minimum, maximum, and average value of $\gamma_{\mathrm{est}}$ for human-only control and shared control during all Monte Carlo simulations. The human-only system amplifies the disturbance with a ratio significantly greater than one and is therefore unstable. In contrast, the nominal shared control reduces the amplification with $\gamma_{\mathrm{est}}\le1$, thereby ensuring stochastic $\mathcal{L}_2$ string stability.

To illustrate the controller's effect on mode tracking, we report a representative run using the same TD-based true mode $\eta(t)$ for both schemes. Fig.~\ref{fig:mode_evolution1} shows $\eta(t)$ under human-only control. Fig.~\ref{fig:mode_evolution2} gives the true mode $\eta(t)$ and the estimated mode $\hat{\eta}(t)$ under nominal shared control. The high task difficulty phase is prolonged when no automated assistance is applied, while it becomes shorter under nominal shared control. This indicates that the nominal controller stabilizes disturbances and shortens high task difficulty phases, thereby reducing driver workload.

\begin{table}[t]
\centering
\begin{threeparttable}
\caption{Stochastic $\mathcal{L}_2$ string stability under 100 Monte Carlo runs.}
\label{tab:L2norm}
\begin{tabularx}{0.85\linewidth}{@{\hspace{6pt}} l @{\extracolsep{\fill}} c c c c l @{\hspace{6pt}}}
\toprule
\multirow{2}{*}{Control scheme} & \multicolumn{4}{c}{$\gamma_{\mathrm{est}}$ statistics} & \multirow{2}{*}{Stability evaluation} \\
\cmidrule(l{-5pt}r{0pt}){2-5}
& Mean & Max & Min & Variance &  \\
\midrule
Human-only             & 1.5853 & 1.6600 & 1.4116 & 2.6$\times 10^{-3}$ & Unstable ($> 1$) \\
Nominal shared control & 0.8572 & 0.8613 & 0.8554 & 2.4$\times 10^{-6}$ & Stable ($\le 1$) \\
\bottomrule
\end{tabularx}
\end{threeparttable}
\end{table}

\begin{figure}[t]
    \centering
    \subfigure[Human-only: true mode $\eta(t)$]{
        \includegraphics[width=0.45\linewidth]{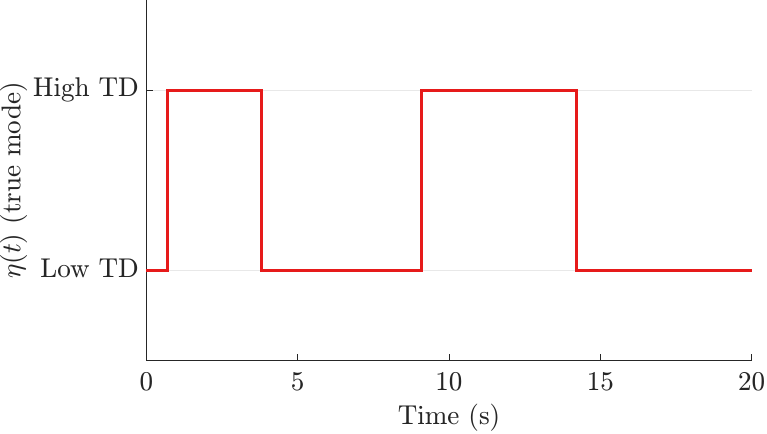}
        \label{fig:mode_evolution1}
    }\quad
    \subfigure[Nominal shared: $\eta(t)$ and $\hat{\eta}(t)$]{
        \includegraphics[width=0.45\linewidth]{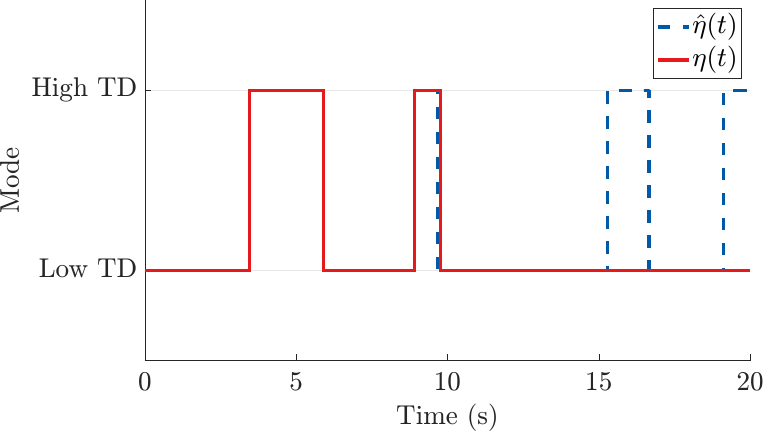}
        \label{fig:mode_evolution2}
    }
    \caption{Mode evolution under the two schemes. The left panel shows the TD-based true mode $\eta(t)$ when no automated assistance is applied. The right panel plots the true mode and the observed mode produced by the estimator under nominal shared control. With AV assistance, the high TD duration is shortened, which implies lower driving difficulties for human drivers.}
\end{figure}

Fig.~\ref{fig:P1_human_nominal} compares the ego-vehicle's acceleration, speed, and spacing profiles under the two control schemes, and the vertical dashed lines indicate the lane-changing time. Under human-only control, the driver completes the lane change at $t=7.51$ s. With shared control, the maneuver finishes earlier at $t=4.59$ s. As shown in Fig.~\ref{fig:P1_human_acc}, the human driver applies a relatively mild deceleration around $t=3$ s to $4$ s, which delays gap acceptance and induces oscillations in the follower vehicle's speed and spacing (Figs.~\ref{fig:P1_human_speed}, \ref{fig:P1_human_gaps}). By contrast, under nominal shared control in Figs.~\ref{fig:P1_nominal_acc}–\ref{fig:P1_nominal_gaps}, the ego-vehicle executes a sharper but well-regulated deceleration, reacting more sensitively to the leader vehicle's motion while simultaneously accounting for the follower vehicle's state. This coordinated response prevents unnecessary speed loss and avoids large oscillations in the follower vehicle, even though the observed mode $\hat{\eta}(t)$ does not perfectly coincide with the true mode $\eta(t)$. Overall, the nominal shared controller enables faster gap acceptance while maintaining stability.
\begin{figure}[t]
    \centering
    Human-only control\\
    \subfigure[Acceleration $a$]{
        \includegraphics[width=0.31\linewidth]{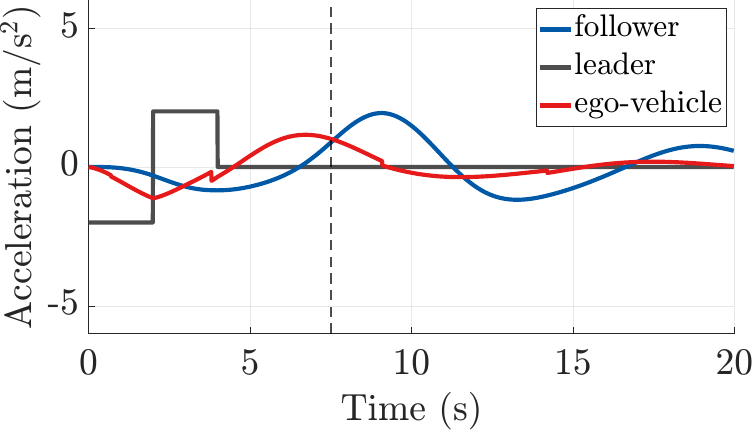}
        \label{fig:P1_human_acc}
    }
    \subfigure[Speed $v$]{
        \includegraphics[width=0.31\linewidth]{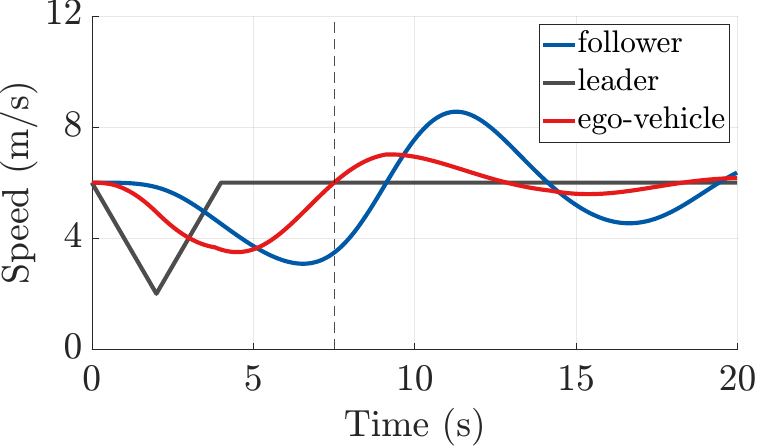}
        \label{fig:P1_human_speed}
    }
    \subfigure[Gap $s$]{
        \includegraphics[width=0.31\linewidth]{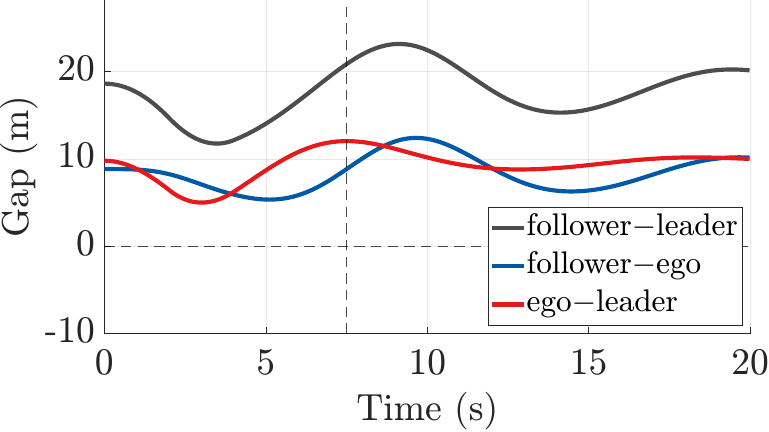}
        \label{fig:P1_human_gaps}
    }\\
    Nominal shared control \\
    \subfigure[Acceleration $a$]{
        \includegraphics[width=0.31\linewidth]{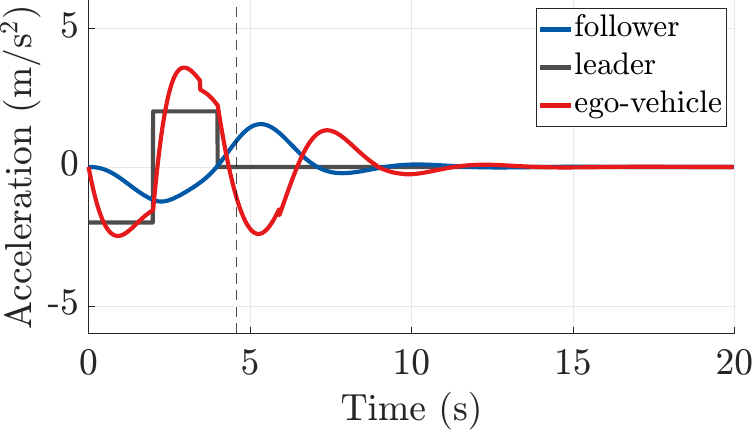}
        \label{fig:P1_nominal_acc}
    }
    \subfigure[Speed $v$]{
        \includegraphics[width=0.31\linewidth]{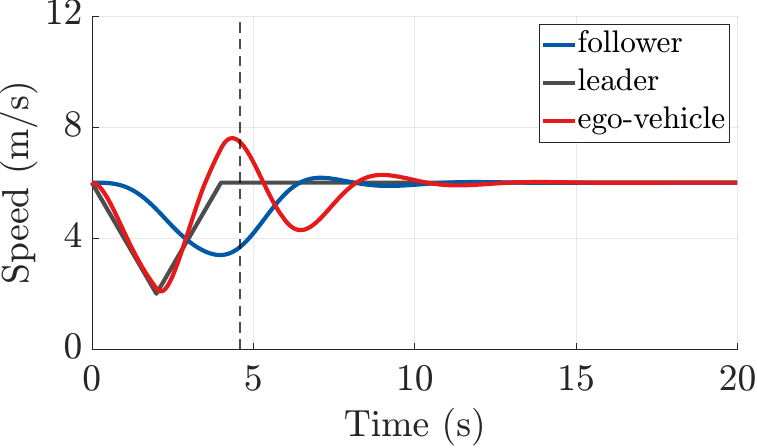}
        \label{fig:P1_nominal_speed}
    }
    \subfigure[Gap $s$]{
        \includegraphics[width=0.31\linewidth]{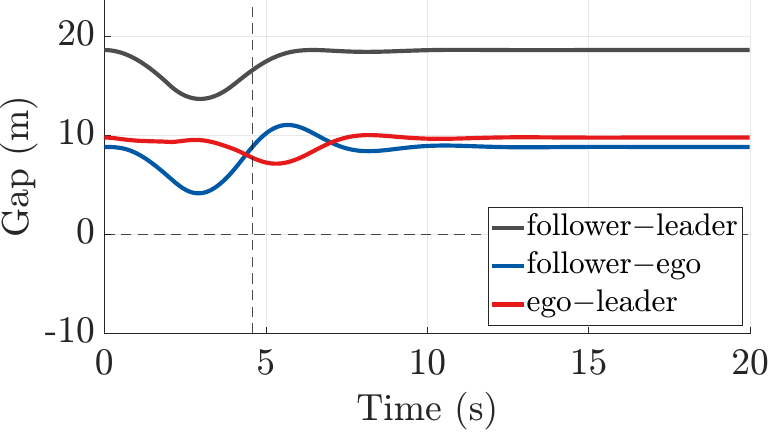}
        \label{fig:P1_nominal_gaps}
    }
    \caption{Trajectories of vehicles under human-only control and nominal shared control.}
    \label{fig:P1_human_nominal}
\end{figure}

\subsection{Compare MIC vs nominal controller}
To illustrate the effect of minimal intervention, we compare the proposed MIC with the nominal controller under the same lane-changing scenario. We set $\beta=2.0$ for the MIC as a representative case. The MIC gains are obtained by solving the LMIs in Theorem \ref{thm:mic_lmi}. This procedure yields the LMI-based bound $\gamma_0$, defined as the minimum feasible value of $\gamma$ under the associated matrix inequalities. Specifically,
\begin{align}
\label{eq:gamma0_def}
& \gamma_{0}= \min_{ P_{ik}> 0,\,K_{\av,k},\,D_{\av,k}}    \gamma \\
\text{s.t.} \quad&
\mathcal{LMI}\!\big(P_{ik},K_{\av,k},D_{\av,k};\beta\big)\ < 0, \notag
\end{align}
where $\mathcal{LMI}(\cdot)< 0$ collects the feasibility conditions stated in Theorem~\ref{thm:mic_lmi}.

Fig.~\ref{fig:P1_MIC_2} shows the simulated vehicle trajectories. Compared with the nominal controller in Figs.~\ref{fig:P1_nominal_acc}–\ref{fig:P1_nominal_gaps}, which completes the maneuver at about $t=4.59$ s, the MIC finishes slightly later at $t=5.98$ s, reflecting its more conservative reliance on automated assistance. The ego-vehicle's acceleration profile in Fig.~\ref{fig:P1_MIC_acc} is noticeably smoother, with the braking peak at $t=3$ s to $4$ s reduced to avoid abrupt deceleration. As a result, the ego-vehicle's speed in Fig.~\ref{fig:P1_MIC_speed} recovers more steadily, stabilizing near $6~\text{m/s}$ by $t=7$ s. For the follower vehicle, oscillations in both speed and gap remain limited (Figs.~\ref{fig:P1_MIC_speed}, \ref{fig:P1_MIC_gaps}), showing that MIC effectively balances responsiveness to the leader vehicle with disturbance mitigation for the follower vehicle and maintains comparable stability characteristics relative to the nominal controller.

\begin{figure}[t]
    \centering
    \subfigure[Acceleration $a$]{
        \includegraphics[width=0.31\linewidth]{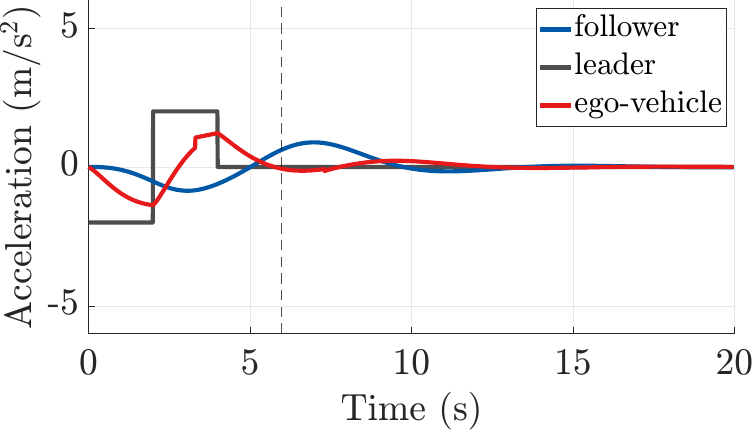}
        \label{fig:P1_MIC_acc}
    }
    \subfigure[Speed $v$]{
        \includegraphics[width=0.31\linewidth]{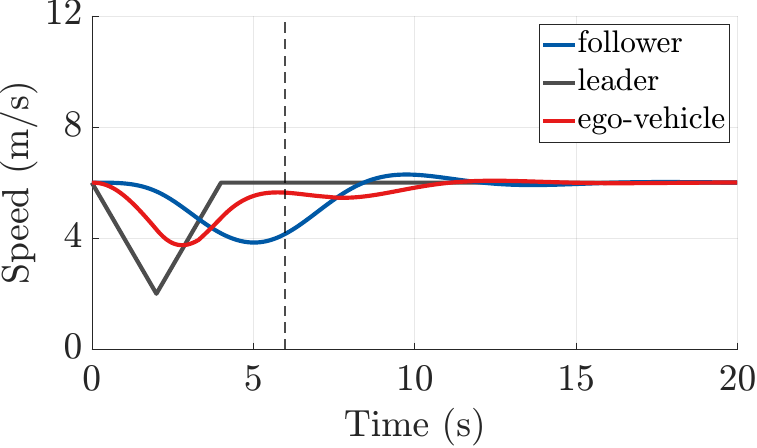}
        \label{fig:P1_MIC_speed}
    }
    \subfigure[Gap $s$]{
        \includegraphics[width=0.31\linewidth]{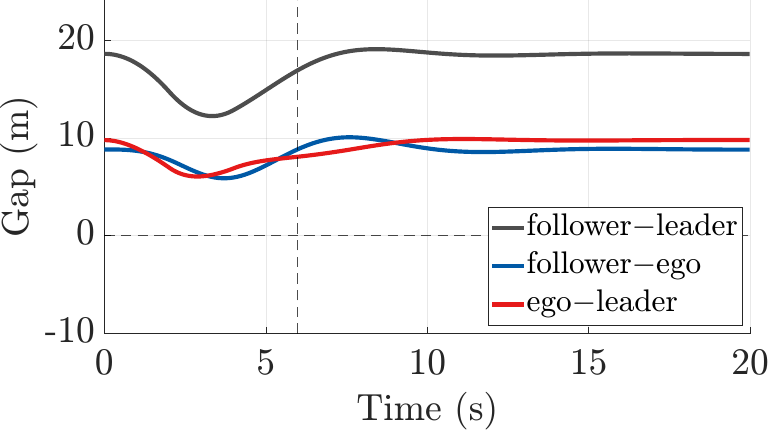}
        \label{fig:P1_MIC_gaps}
    }
    \caption{Trajectories of vehicles under shared control with MIC.}
    \label{fig:P1_MIC_2}
\end{figure}

Fig.~\ref{fig:delta_u_compare} compares the control effort of the nominal controller and MIC. Under the nominal controller, the difference signal $\Delta u(t)=u_{\human}(t)-u_{\av}(t)$ in Fig.~\ref{fig:nominal_delta_u} reaches relatively large values, with peaks approaching $\pm 4~\text{m/s}^2$ around $t=2$ s and $t=4$ s. This indicates that the automated system applies considerable effort to stabilize the system. 
As seen in Fig.~\ref{fig:nominal_inputs}, $u_{\av}$ exhibits strong peaks in the opposite direction of $u_{\human}$, and the combined input $u(t)=u_{\human}(t)+u_{\av}(t)$ correspondingly produces pronounced peaks at similar times.
By contrast, when MIC is applied, as shown in Figs.~\ref{fig:MIC_delta_u} and \ref{fig:MIC_inputs}, the difference $\Delta u(t)$ becomes much smaller. Around $t=4$ s, its minimum reaches roughly $-2~\text{m/s}^2$, while the total input $u(t)$ also exhibits reduced peaks with smoother temporal variations, as $u_{\av}$ remains moderate and more aligned with $u_{\human}$. 
These results demonstrate that MIC not only reduces automated intervention but also improves coordination between human and automation, thereby preserving driver authority while maintaining acceptable stability performance.

\begin{figure}[t]
    \centering
    Nominal controller\\
    \subfigure[$\Delta u = u_\human-u_\av$]{
        \includegraphics[width=0.31\linewidth]{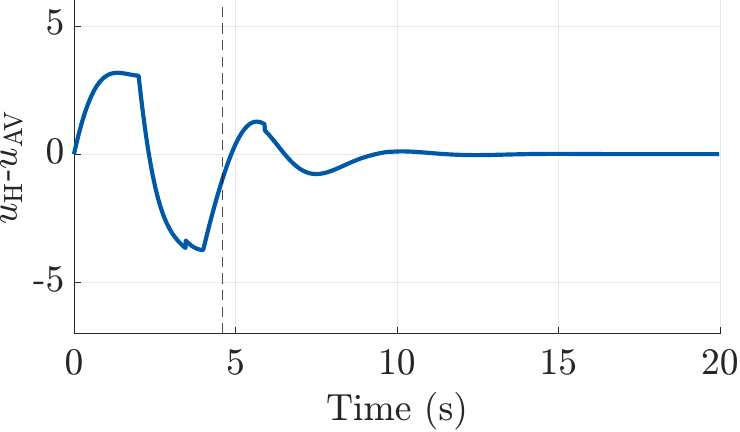}
        \label{fig:nominal_delta_u}
    }
    \qquad
    \subfigure[$u = u_\human+u_\av$]{
        \includegraphics[width=0.31\linewidth]{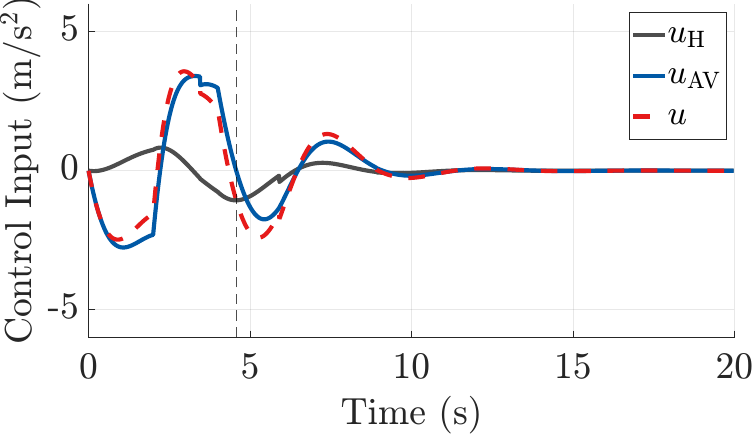}
        \label{fig:nominal_inputs}
    }
    \\
    Minimal intervention controller\\
    \subfigure[$\Delta u = u_\human-u_\av$]{
        \includegraphics[width=0.31\linewidth]{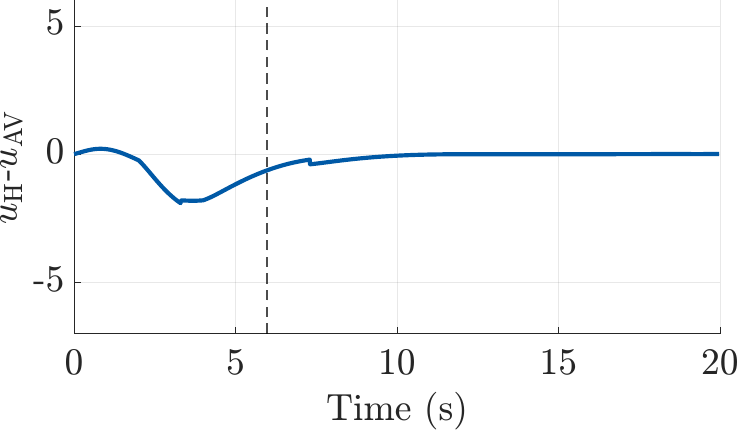}
        \label{fig:MIC_delta_u}
    }
    \qquad  
    \subfigure[$u = u_\human+u_\av$]{
        \includegraphics[width=0.31\linewidth]{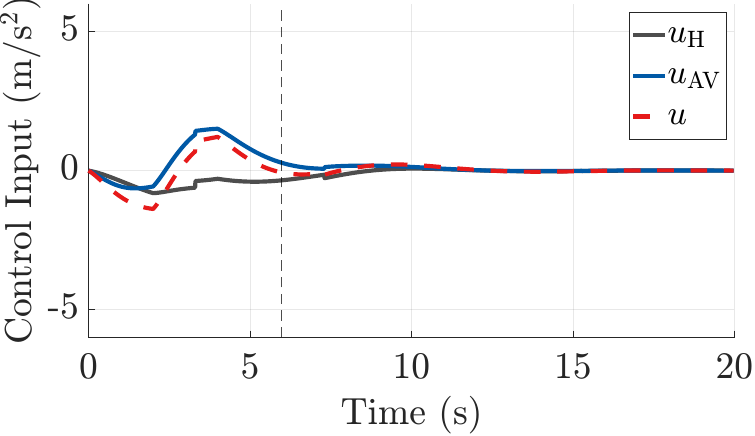}
        \label{fig:MIC_inputs}
    }
    \caption{Comparison of input difference $\Delta u =  u_\human-u_\av$ and control inputs $u = u_\human+u_\av$ between nominal and minimal intervention controllers.}
    \label{fig:delta_u_compare}
\end{figure}

\subsection{Sensitivity analysis of MIC with $\beta$}
The MIC introduces the effort weight $\beta$ to balance disturbance attenuation against automated effort, thereby adjusting the degree of driver authority. To analyze the trade-off, we run simulations over a uniform grid $\beta_{\text{grid}}=\{0.5,1.0,\ldots,5.0\}$. For each $\beta\in\beta_{\text{grid}}$, the MIC gains are synthesized by solving Eq.~\eqref{eq:gamma0_def}.

The relative automated intervention effort is quantified by an intervention ratio defined as:
\begin{align}    
r_{\mathrm{int}}&=\frac{E_\av}{E_\av+E_\human}, \\
E_\av=&\| u_\av \|_{\mathcal{L}_2(0,T)},\quad E_\human=\| u_\human \|_{\mathcal{L}_2(0,T)}.
\end{align}
Moreover, we define comfort indicators $\text{RMS}_{a\mathrm{\egov}}$ and $\text{RMS}_{a\mathrm{\followv}}$ to evaluate ride comfort of the ego-vehicle and the follower vehicle:
\begin{align}
    \text{RMS}_{a\mathrm{\egov}} = \sqrt{\frac{1}{T}\int_0^T a_{\egov}^2(t)\,dt}, \quad \text{RMS}_{a\mathrm{\followv}} = \sqrt{\frac{1}{T}\int_0^T a_{\followv}^2(t)\,dt}.
\end{align}
A smaller $\text{RMS}_a$ indicates smoother and more comfortable trajectories, while a larger value reflects stronger accelerations and reduced ride comfort. 

For each $\beta$, we run 100 Monte Carlo trials with randomized mode paths. Figures~\ref{fig:rint_beta}–\ref{fig:merge_beta} plot averaged $r_{\mathrm{int}}$, $\gamma_0$, $\gamma_{\mathrm{est}}$, $\text{RMS}_{a\mathrm{\egov}}$, $\text{RMS}_{a\mathrm{\followv}}$ and the lane-changing completion time $t_{\mathrm{LC}}$ of 100 experiments against different $\beta$ values. 
In Fig.~\ref{fig:rint_beta}, as $\beta$ increases, the automated effort ratio $r_{\mathrm{int}}$ decreases monotonically from about $0.75$ at $\beta=0.5$ to below $0.43$ at $\beta=5.0$, indicating that the control authority gradually reallocates toward the human driver. Compared with the nominal controller, which yields $r_{\mathrm{int}}$ around $0.76$, the MIC consistently produces smaller intervention ratios, with the stronger reduction at larger $\beta$.
The theoretical LMI bound $\gamma_0$ increases consistently from around the nominal level at $\beta = 1.0$ to more than twice that value at $\beta=5.0$ in Fig.~\ref{fig:gamma0_beta}, indicating that stronger penalties on automated effort lead to weaker theoretical stability margins. In comparison, the nominal controller maintains a tight bound with $\gamma_0$ below one. 
From Fig.~\ref{fig:gamma_est_beta}, $\gamma_{\mathrm{est}}$ shows a non-monotonic trend with respect to $\beta$. At $\beta=1.0$, $\gamma_{\mathrm{est}}$ reaches its minimum of about $0.80$, slightly below the nominal controller value of $0.85$. As $\beta$ increases to $2.0$ and $3.0$, $\gamma_{\mathrm{est}}$ stays below the nominal baseline (around $0.82$–$0.88$), still outperforming the nominal controller while allowing reduced automation effort. At $\beta=5.0$, however, $\gamma_{\mathrm{est}}$ surpasses the nominal baseline, indicating a degradation of stability performance despite further reductions in automated intervention.

As shown in Fig.~\ref{fig:rms_ae}, $\text{RMS}_{a\mathrm{\egov}}$ decreases as $\beta$ increases, dropping from roughly $0.8~\text{m/s}^2$ to below $0.5~\text{m/s}^2$ when $\beta = 2$, and then gradually levels off around $0.4~\text{m/s}^2$ for $\beta\ge 2$. Across the entire range, it remains well below the nominal controller baseline of approximately $1.2~\text{m/s}^2$, demonstrating that penalizing automated effort consistently improves the comfort of the ego-vehicle.
By contrast, Fig.~\ref{fig:rms_af} shows that at small $\beta$ (around $1.0$–$2.0$), $\text{RMS}_{a\mathrm{\followv}}$ is reduced to about $0.38$–$0.40\text{m/s}^2$, significantly lower than the nominal baseline of roughly $0.52~\text{m/s}^2$, indicating improved comfort for the follower vehicle. However, as $\beta$ increases further, $\text{RMS}_{a\mathrm{\followv}}$ starts rising, reaching nearly the nominal level by $\beta=5.0$. This suggests that when automated effort is penalized too strongly, the controller's stabilizing capacity diminishes, and the follower vehicle experiences greater fluctuations.

In Fig.~\ref{fig:merge_beta}, the average lane-changing time $t_{\mathrm{LC}}$ increases monotonically with $\beta$, since reduced automated assistance prolongs maneuver execution. Under the nominal controller, the average completion time is about $4.59$ s, whereas with MIC, it rises to $5.56$ s at $\beta=1.0$. As $\beta$ increases further to $2.0$ and $3.0$, the lane-changing time grows moderately longer (about $6.00$–$6.35$ s), and by $\beta=4.0$ it reaches $6.74$ s. At very high $\beta$, the maneuver may approach the simulation horizon without completing. To summarize, findings in Fig.~\ref{fig:beta_sweep} confirm the inherent trade-off: stronger penalties on automated effort shift control authority toward the driver and improve ego-vehicle ride comfort, but simultaneously weaken disturbance attenuation for the follower vehicle and prolong the maneuver relative to the nominal baseline. Moderate $\beta$ values around $2.0$–$3.0$ achieve the most favorable balance between reduced intervention, comfort, and timely lane-changing.

\begin{figure}[t]
    \centering
    \subfigure[Automated effort ratio $r_{\mathrm{int}}$]{
        \includegraphics[width=0.31\linewidth]{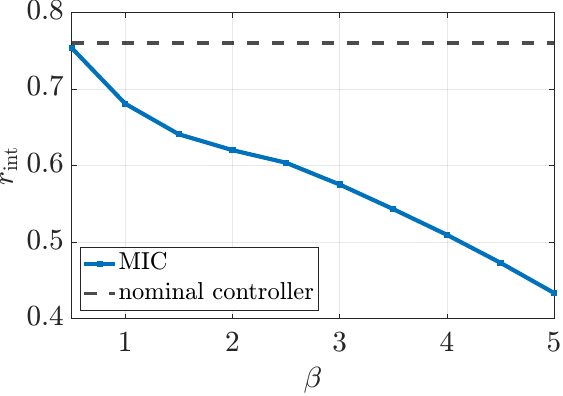}\label{fig:rint_beta}
    }
    \subfigure[Theoretical LMI bound $\gamma_0$]{
    \includegraphics[width=0.31\linewidth]{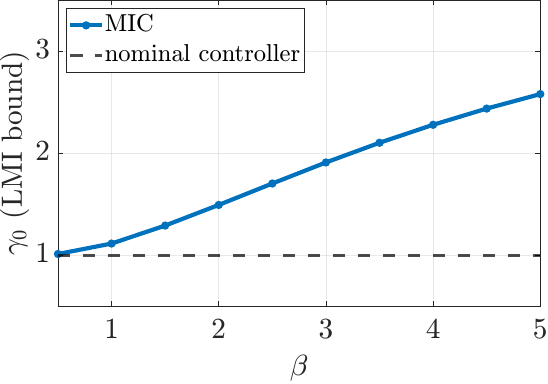}\label{fig:gamma0_beta}
    }
    \subfigure[Empirical gain $\gamma_{\mathrm{est}}$]{
    \includegraphics[width=0.31\linewidth]{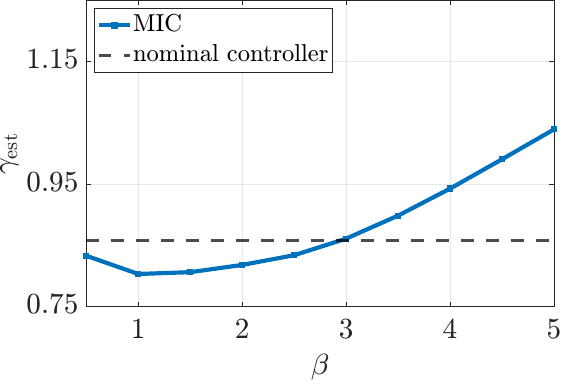}\label{fig:gamma_est_beta}
    }\\
    \subfigure[Ego-vehicle comfort $\text{RMS}_{a\mathrm{\egov}}$]{
    \includegraphics[width=0.31\linewidth]{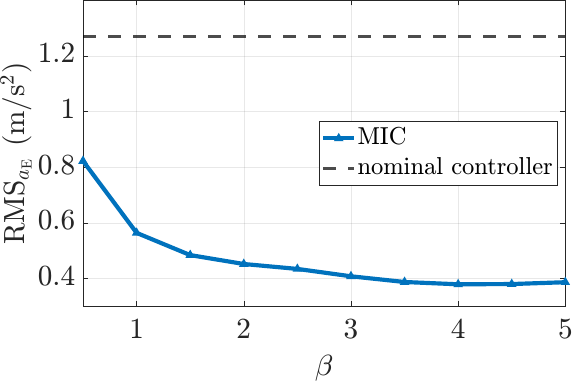}\label{fig:rms_ae}
    }
    \subfigure[Follower vehicle comfort $\text{RMS}_{a\mathrm{\followv}}$]{
    \includegraphics[width=0.31\linewidth]{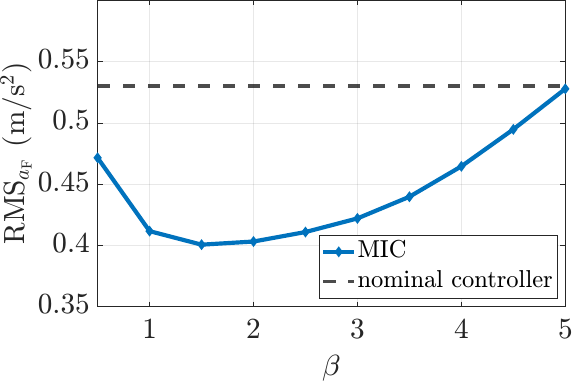}\label{fig:rms_af}
    }
    \subfigure[Lane-changing completion time $t_{\mathrm{LC}}$]{
        \includegraphics[width=0.31\linewidth]{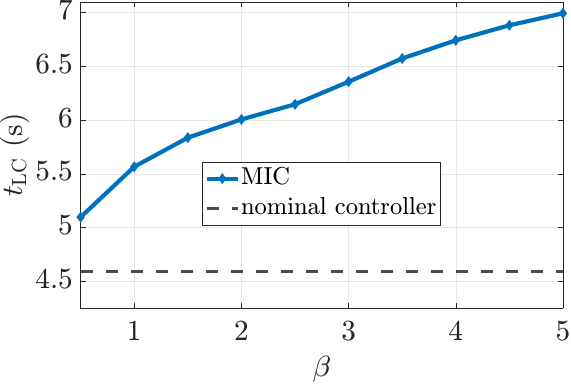}\label{fig:merge_beta}
    }
    \caption{Results under MIC parameter $\beta$. In all plots, solid curves show MIC results and horizontal dashed lines show the nominal controller.}    \label{fig:beta_sweep}
\end{figure}

\subsection{Validation with real-world trajectories}\label{sec:raw_data}

\subsubsection{Human-only control vs. Shared control}
To further assess the proposed controller under realistic conditions, we use a lane-changing instance from the NGSIM I-80 dataset. Specifically, a 20-second window of data is extracted that contains a complete lane-change event, including the trajectories of the ego-vehicle, the leader vehicle, and the follower vehicle.

\begin{figure}[t]
    \centering
    Human-only control\\
    \subfigure[Acceleration $a$]{
        \includegraphics[width=0.31\linewidth]{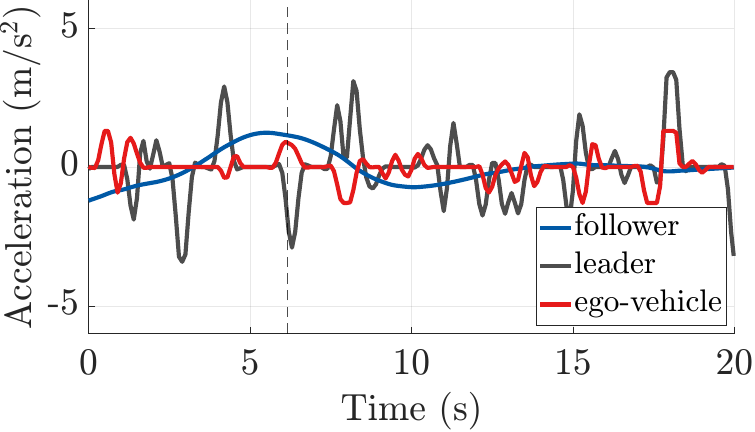}
        \label{fig:N_human_acc}
    }
    \subfigure[Speed $v$]{
        \includegraphics[width=0.31\linewidth]{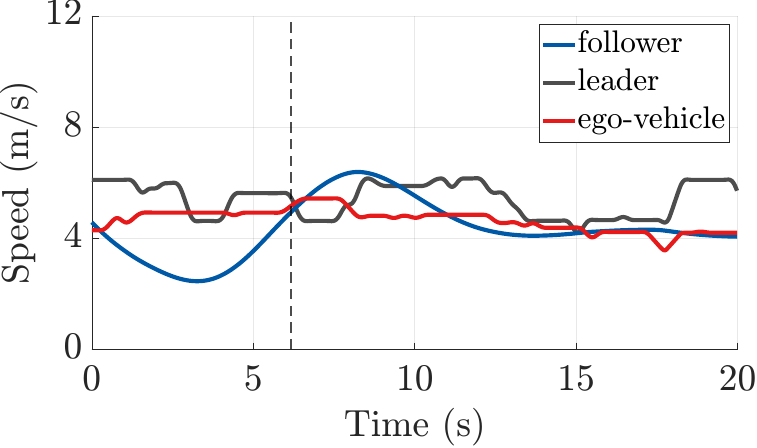}
        \label{fig:N_human_speed}
    }
    \subfigure[Gap $s$]{
        \includegraphics[width=0.31\linewidth]{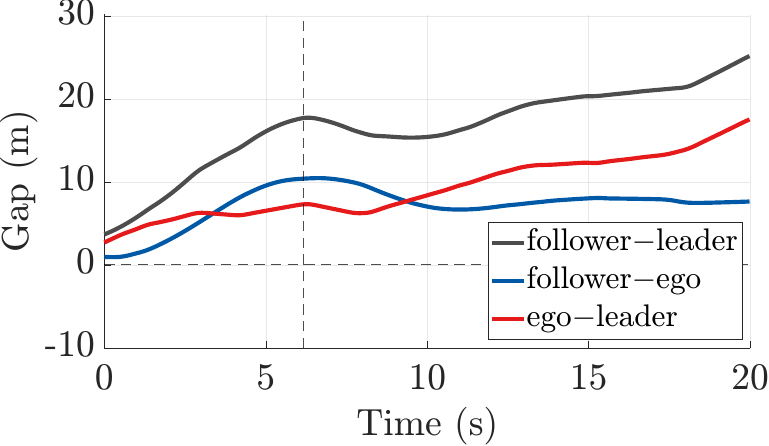}
        \label{fig:N_human_gaps}
    }\\
    Shared control with MIC\\
    \subfigure[Acceleration $a$]{
        \includegraphics[width=0.31\linewidth]{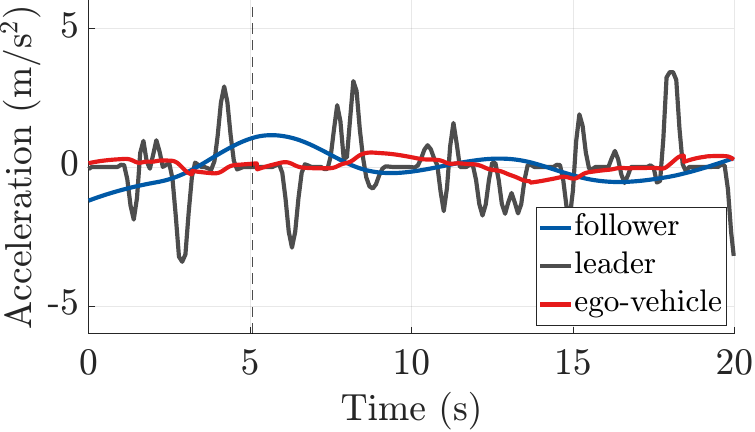}
        \label{fig:N_MIC_acc}
    }
    \subfigure[Speed $v$]{
        \includegraphics[width=0.31\linewidth]{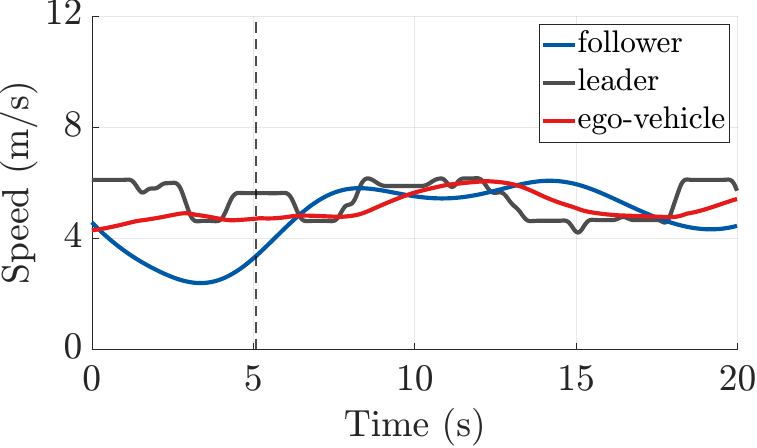}
        \label{fig:N_MIC_speed}
    }
    \subfigure[Gap $s$]{
        \includegraphics[width=0.31\linewidth]{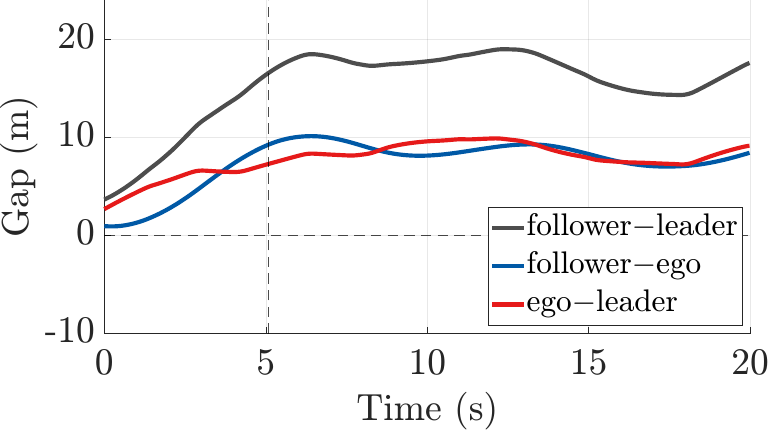}
        \label{fig:N_MIC_gaps}
    }
    \caption{Trajectories of vehicles under human-only control and shared control with MIC (with NGSIM trajectories).}
    \label{fig:N_human_MIC}
\end{figure}

We compare human-only control with the MIC in Fig.~\ref{fig:N_human_MIC}. As shown in Figs.~\ref{fig:N_human_acc}-\ref{fig:N_human_gaps}, with human-only control, the ego-vehicle exhibits pronounced fluctuations in acceleration, with alternating mild braking and re-acceleration to adjust its gap from the follower vehicle. The ego-vehicle completes the lane change at $t=6.17$ s while slightly accelerating to enlarge the spacing from the follower vehicle. 
The follower vehicle exhibits oscillations in both acceleration and speed, with the speed dropping to about $2.5~\text{m/s}$ at $t=3$ s, rising to nearly $6.0~\text{m/s}$ at $t=8$ s, and then falling to about $4~\text{m/s}$. These amplified responses result in an empirical gain of $\gamma_{\mathrm{est}}=2.276$. The comfort indicators are $\text{RMS}_{a\mathrm{\egov}}=0.47~\text{m/s}^2$ and $\text{RMS}_{a\mathrm{\followv}}=0.58~\text{m/s}^2$.
With MIC, the assistive controller helps the ego-vehicle maintain its initial speed and coordinate the gaps to both leader vehicle and follower vehicle, enabling the lane change to be completed earlier at $t=5.08$ s and the automated effort ratio remains moderate with $r_{\mathrm{int}}=0.54$. As shown in Figs.~\ref{fig:N_MIC_acc}-\ref{fig:N_MIC_gaps}, the ego-vehicle makes smoother adjustments in acceleration and velocity. The follower vehicle's speed variation is mitigated, showing one noticeable peak of about $2.5~\text{m/s}$ at $t=3$ s. This results in a smaller empirical gain of $\gamma_{\mathrm{est}}=2.016$. The comfort indicators further confirm the improvement, with $\text{RMS}_{a\mathrm{\egov}}=0.26~\text{m/s}^2$ and $\text{RMS}_{a\mathrm{\followv}}=0.53~\text{m/s}^2$. Although both empirical gains remain above one within the finite simulation horizon, the MIC limits intervention to moderate adjustments, thereby preserving driver authority, ensuring smoother motion, and enabling more efficient lane changing.

\subsubsection{Automation-only control vs. Shared control}
We further validate the effectiveness of the proposed controller using the Third Generation Simulation (TGSIM) dataset collected by moving cameras on the I–90/94 corridor \cite{ammourah2025introduction}. From this dataset, we select lane-changing events in which the ego-vehicle is an SAE Level 2 automated vehicle while the surrounding vehicles (leader vehicle and follower vehicle) are human-driven. Current L2 systems implement automation-only control, where the automated controller takes full charge until the driver disengages, so human and automation cannot operate simultaneously. This makes L2 automated system a suitable baseline for comparison with our shared control framework, which allows shared human and automation inputs.

We first calibrate the human driving model from HV lane-change cases. Similar to the NGSIM analysis, the ego-vehicle is represented by a mode-dependent OVM with Markov mode switching governed by a TD-based thresholding rule, and the transition rates are estimated from multiple trajectories. 
The calibrated parameters and TD thresholds are summarized in Table~\ref{tab:tgsim_all}. For the follower vehicle, we adopt the same leader-selecting OVM formulation as in Eq.~\eqref{eq:fv_model}, yielding calibrated parameters of $v_{\max}=25.00~\text{m/s}$, $s_{\mathrm{st}}=5.00~\text{m}$, $s_{\mathrm{go}}=25.00~\text{m}$, $a=0.12$, and $b=0.11$.

For simulation, we replay the leader vehicle trajectory directly from the TGSIM data, while the follower vehicle is simulated using the calibrated OVM model. The ego-vehicle is modeled in two ways: (i) as an L2 automated vehicle, and (ii) as a shared control vehicle that combines the calibrated human OVM with the minimal intervention controller (MIC). In the shared control case, the estimator applies $\hat{\eta}(t)$ with $\alpha=0.05$ and $q=0.02~\text{s}^{-1}$ to select the feedback pair $(K_{\av,k},D_{\av,k})$. The lane-changing condition is declared satisfied when both the ego–leader gap $s_{\egov\leadv}$ and the ego-follower gap $s_{\followv\egov}$ exceed the minimum thresholds extracted from the dataset, with calibrated thresholds $s_{\text{rear,thr}}=10.9$ m and $s_{\text{front,thr}}=12.1$ m.

\begin{table}[t]
\centering
\begin{threeparttable}
\caption{Calibrated driver parameters from the TGSIM dataset}
\label{tab:tgsim_all}
\begin{tabularx}{0.9\linewidth}{@{\hspace{6pt}} l @{\extracolsep{\fill}} c c c l @{\hspace{6pt}}}
\toprule
Category & Symbol & Value & Unit & Description \\
\midrule
TD parameters & $T_{\text{des}}$ & 1.18 & s   & Desired headway \\
    & $\delta$         & 0.25 & -  & Risk sensitivity \\
    & $\zeta$         & 1.0 & -  & Exponent for TD scaling \\
\midrule
Markov transition rates & $\lambda_{11}$ & -0.0171 & s$^{-1}$ & Leaving rate from low TD \\
    & $\lambda_{12}$ & 0.0171 & s$^{-1}$ & Switching $1 \to 2$ \\
    & $\lambda_{21}$ & 0.0771 & s$^{-1}$ & Switching $2 \to 1$ \\
    & $\lambda_{22}$ & -0.0771 & s$^{-1}$ & Leaving rate from high TD \\
\midrule
Mode-dependent OVM & $a_1$ & 0.20 & s$^{-1}$  & Sensitivity to desired speed ($\eta=1$) \\
    & $b_1$ & 0.10 & s$^{-1}$  & Sensitivity to relative speed ($\eta=1$) \\
    & $a_2$ & 0.11 & s$^{-1}$  & Sensitivity to desired speed ($\eta=2$) \\
    & $b_2$ & 0.18 & s$^{-1}$  & Sensitivity to relative speed ($\eta=2$) \\
    & $s_{\mathrm{st}}$ & 5.00 & m   & Stopping distance ($\eta=1,2$) \\
    & $s_{\mathrm{go}}$ & 35.00 & m   & Go distance ($\eta=1,2$)\\
    & $v_{\max}$ & 20.00 & m/s & Maximum desired speed ($\eta=1,2$)\\
\bottomrule
\end{tabularx}
\begin{tablenotes}
\footnotesize
\item TD and Markov rates are calibrated from TGSIM lane-change trajectories.  
\item OVM parameters are estimated by forward-simulation fitting, showing that when $\eta = 1$, driver is more sensitive to the desired speed, while in $\eta = 2$ driver is sensitive to the relative speed.
\end{tablenotes}
\end{threeparttable}
\end{table}

\begin{figure}[t]
    \centering
    Automation-only control\\
    \subfigure[Acceleration $a$]{
        \includegraphics[width=0.31\linewidth]{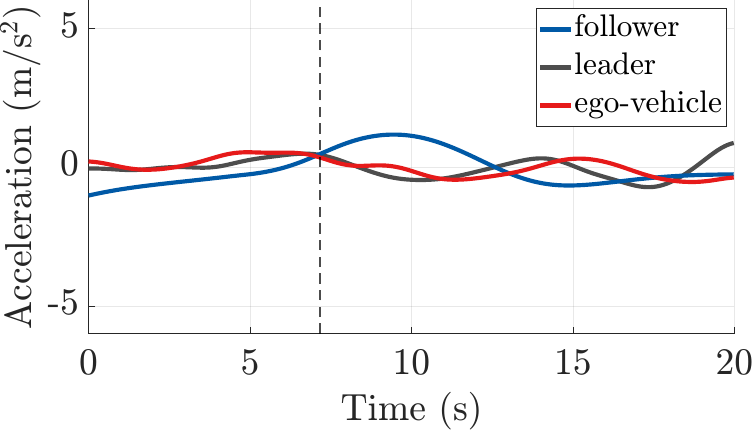}
        \label{fig:TG_ACC_case1_acc}
    }
    \subfigure[Speed $v$]{
        \includegraphics[width=0.31\linewidth]{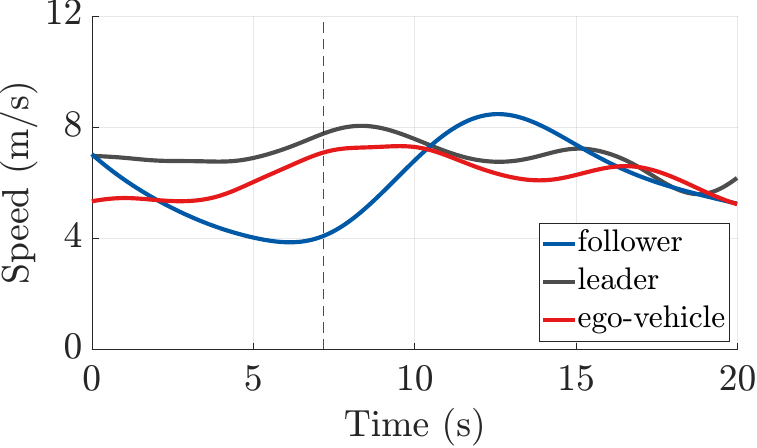}
        \label{fig:TG_ACC_case1_speed}
    }
    \subfigure[Gap $s$]{
        \includegraphics[width=0.31\linewidth]{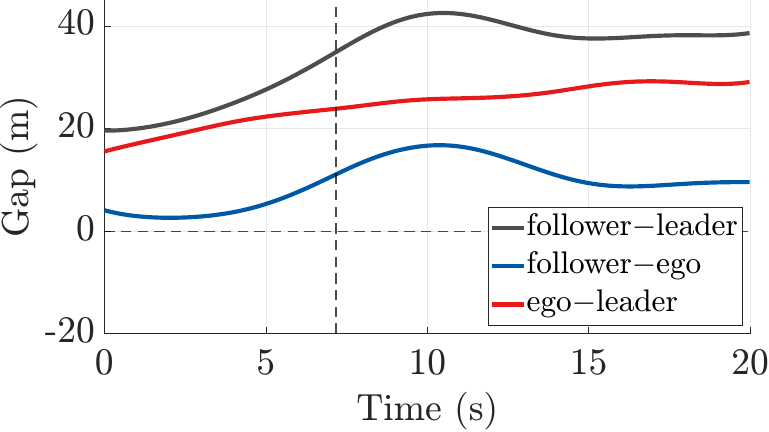}
        \label{fig:TG_ACC_case1_gaps}
    }\\
    Shared control with MIC\\
    \subfigure[Acceleration $a$]{
        \includegraphics[width=0.31\linewidth]{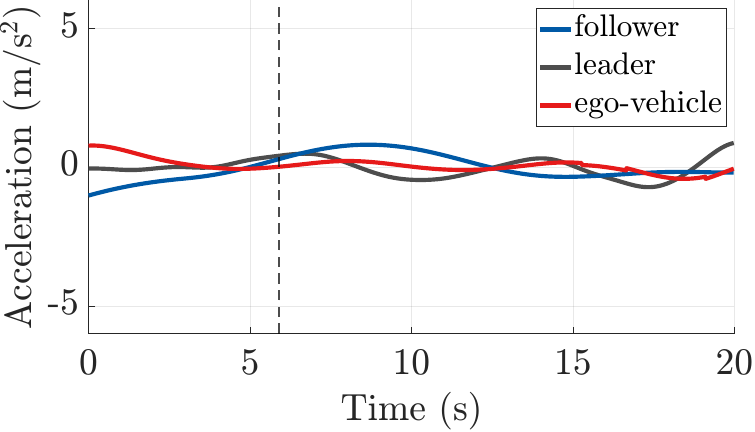}
        \label{fig:TG_MIC_case1_acc}
    }
    \subfigure[Speed $v$]{
        \includegraphics[width=0.31\linewidth]{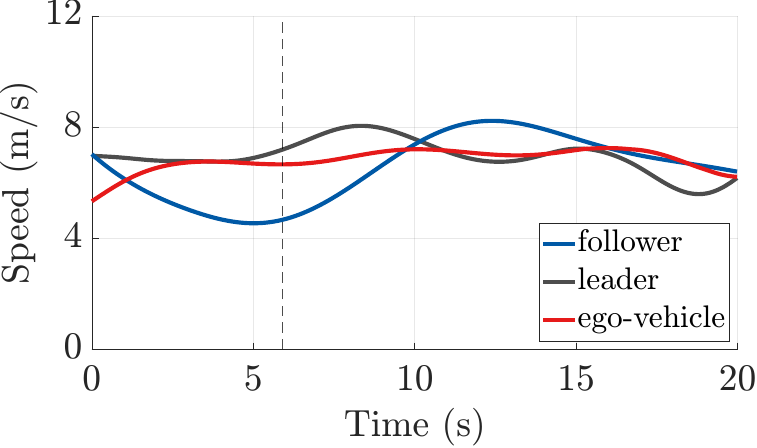}
        \label{fig:TG_MIC_case1_speed}
    }
    \subfigure[Gap $s$]{
        \includegraphics[width=0.31\linewidth]{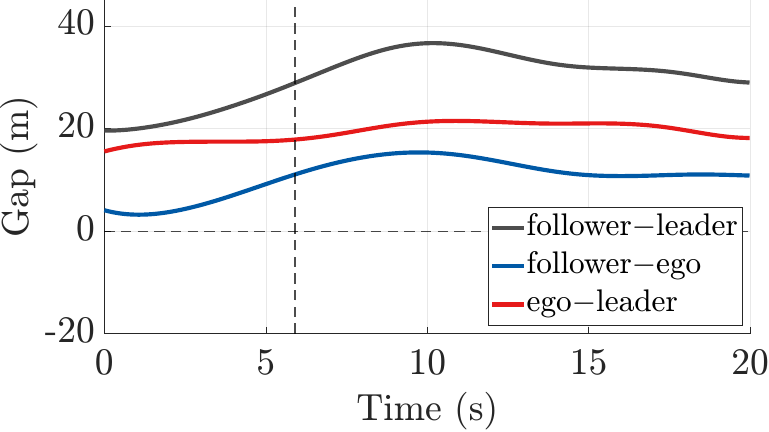}
        \label{fig:TG_MIC_case1_gaps}
    }
    \caption{Trajectories of vehicles under automation-only control and shared control with MIC (with TGSIM Case 1).}
    \label{fig:ACC_MIC_case1}
\end{figure}

\begin{figure}[t]
    \centering
    Automation-only control\\
    \subfigure[Acceleration $a$]{
        \includegraphics[width=0.31\linewidth]{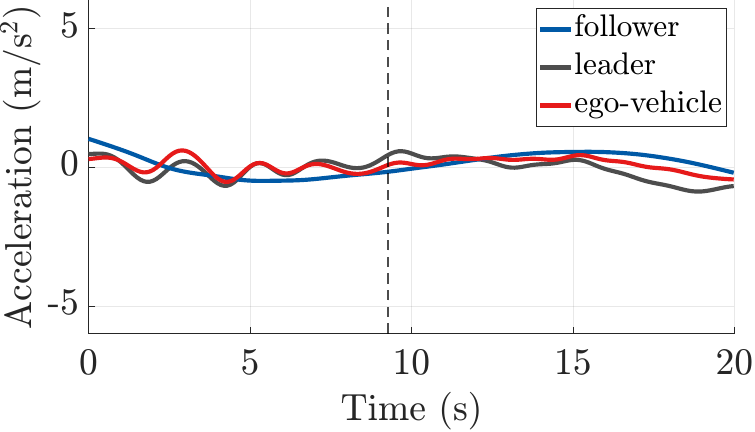}
        \label{fig:TG_ACC_case2_acc}
    }
    \subfigure[Speed $v$]{
        \includegraphics[width=0.31\linewidth]{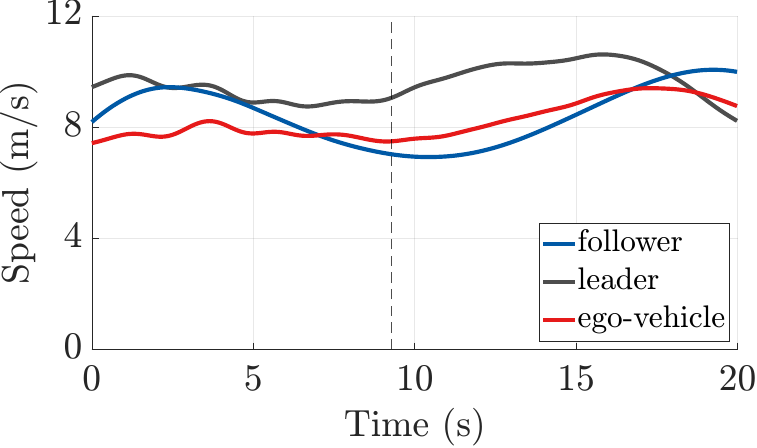}
        \label{fig:TG_ACC_case2_speed}
    }
    \subfigure[Gap $s$]{
        \includegraphics[width=0.31\linewidth]{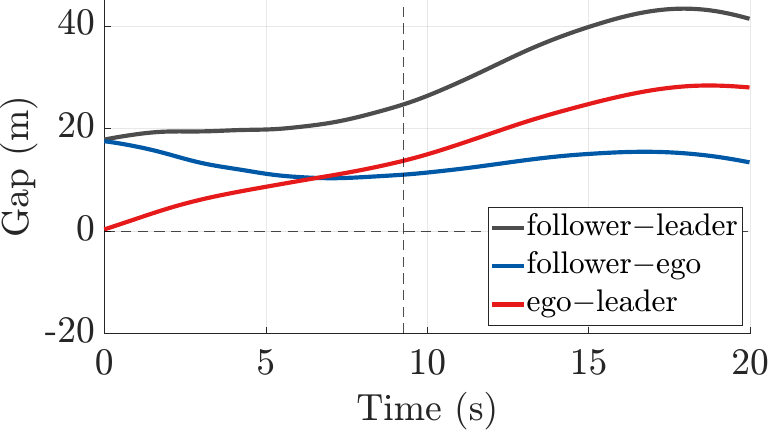}
        \label{fig:TG_ACC_case2_gaps}
    }\\
    Shared control with MIC\\
    \subfigure[Acceleration $a$]{
        \includegraphics[width=0.31\linewidth]{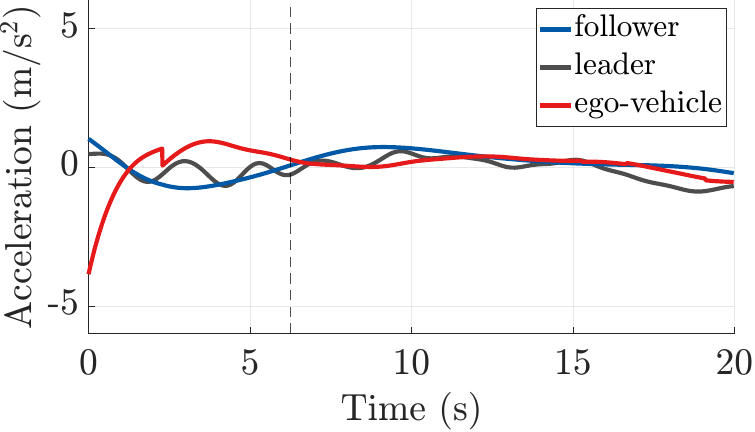}
        \label{fig:TG_MIC_case2_acc}
    }
    \subfigure[Speed $v$]{
        \includegraphics[width=0.31\linewidth]{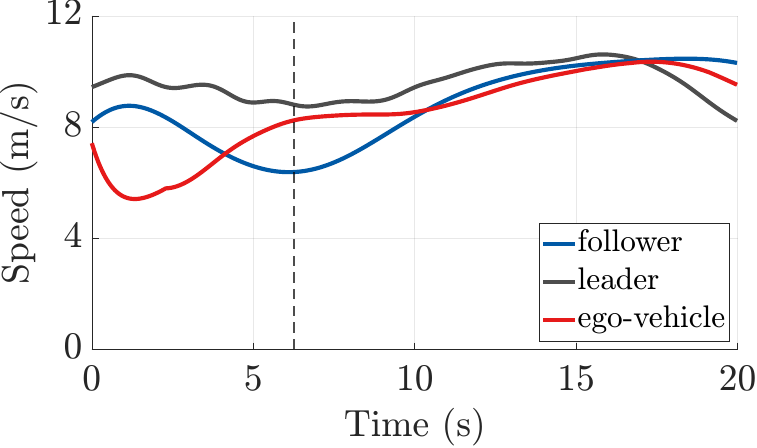}
        \label{fig:TG_MIC_case2_speed}
    }
    \subfigure[Gap $s$]{
        \includegraphics[width=0.31\linewidth]{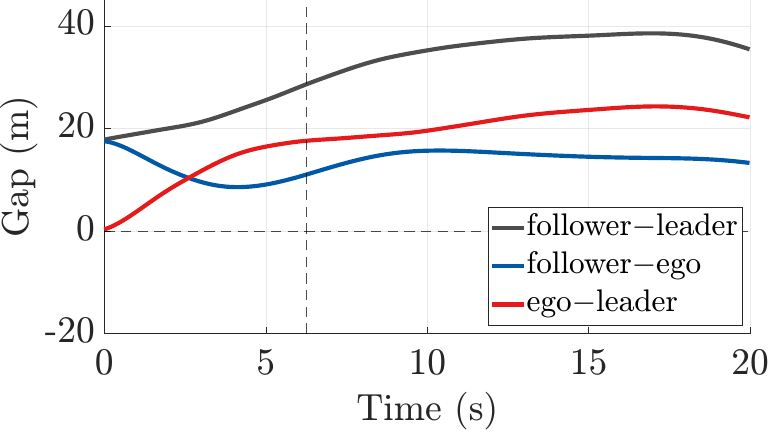}
        \label{fig:TG_MIC_case2_gaps}
    }
    \caption{Trajectories of vehicles under automation-only control and shared control with MIC (with TGSIM Case 2).}
    \label{fig:ACC_MIC_case2}
\end{figure}

We run two cases with different initial configurations to examine the effect of shared control relative to automation-only control. In Case 1 in Fig.~\ref{fig:ACC_MIC_case1}, the ego-vehicle is closer to the follower vehicle initially. The ego-vehicle's initial speed is at approximately $7~\text{m/s}$, equal to the leader vehicle's initial speed, while the follower vehicle starts at a lower speed of about $5~\text{m/s}$. Under automation-only control in Figs.~\ref{fig:TG_ACC_case1_acc}–\ref{fig:TG_ACC_case1_gaps}, the L2 automated system regulates the ego-vehicle primarily by tracking the leader vehicle, with the follower vehicle simply following behind, leading to a relatively stable and smooth speed profile, with the gain $\gamma_{\mathrm{est}}=0.8612$. However, this behavior does not explicitly account for gap coordination with the follower vehicle, so the lane change is delayed, with a completion time of $t=7.17$ s. By contrast, the MIC introduces an early acceleration strategy at the beginning of the maneuver, which enlarges the spacing to the follower vehicle, as shown in Figs.~\ref{fig:TG_MIC_case1_acc}-\ref{fig:TG_MIC_case1_gaps}. This cooperative adjustment enables the ego-vehicle to reach the lane-changing condition much earlier, with a lane-changing time of $t=5.90$ s. The stability performance is slightly weaker than automation-only, with $\gamma_{\mathrm{est}}=0.9340$. Nevertheless, the automated effort ratio is $r_\mathrm{int}=0.51$, indicating that MIC limits intervention and preserves driver authority. 

In Case 2 in Fig.~\ref{fig:ACC_MIC_case2}, the ego-vehicle is initially closer to the leader vehicle, and its initial speed is slightly lower than both leader vehicle and follower vehicle. Under automation-only control, as shown in Figs.~\ref{fig:TG_ACC_case2_acc}–\ref{fig:TG_ACC_case2_gaps}, the ego-vehicle mainly follows the leader vehicle's motion, maintaining similar accelerations, and the follower vehicle also follows the ego-vehicle's motion, so the speed profile is also relatively stable with an empirical gain of $\gamma_{\mathrm{est}}=0.6476$, but the maneuver is prolonged and the lane change completes only at $t=9.27$ s.
With MIC, the ego-vehicle initiates an early deceleration, as shown in Fig.~\ref{fig:TG_MIC_case2_acc}, which enlarges the spacing to the leader vehicle in Fig.~\ref{fig:TG_MIC_case2_gaps} and creates more favorable lane-changing conditions. This adjustment shortens the lane-changing time to $t=6.26$ s. However, the stability performance is decreased, with $\gamma_{\mathrm{est}}=0.9363$. The automated effort ratio of $r_{\mathrm{int}}=0.60$ indicates that MIC still limits automation involvement, achieving faster execution but at the cost of increased fluctuations.

\section{Conclusion}
In this paper, we designed shared controller for lane changes. We modeled the lane-changing process under human–automation shared control as a Markov jump linear system with imperfect mode observation, and established a corresponding stochastic $\mathcal{L}_2$ string stability criterion. A nominal controller was synthesized via LMIs to stabilize the longitudinal dynamics under stochastic human variability, converting an unstable human-only baseline into a stochastically $\mathcal{L}_2$ string stable system. Building on this, we developed a minimal intervention controller that augments the performance output with an automated effort penalty, allowing the trade-off between disturbance attenuation and automated effort to be explicitly tuned by a weight $\beta$. Using calibrated human models from the NGSIM dataset, simulations confirmed that the nominal controller ensures stochastic $\mathcal{L}_2$ string stability, while the MIC introduces a tunable balance. A sensitivity analysis relative to the nominal case revealed a clear trade-off among stability, efficiency, and driver authority: increasing $\beta$ reduces automated effort and improves comfort but weakens stability and prolongs lane-changing time. With moderate $\beta$, the MIC significantly suppresses disturbance propagation while reducing automated effort and improving ride comfort, thereby preserving driver authority compared with the human-only baseline. Validations with TGSIM lane-change cases further showed that, compared with SAE Level 2 automated vehicles, the MIC enables earlier and more efficient lane changes and preserves driver authority, while exhibiting slightly reduced stability. Overall, these findings highlight the potential of shared control strategies to balance stability, efficiency, and driver acceptance in stochastic traffic environments. Future research could extend this framework to multi-vehicle interactions and incorporate more sophisticated human behavior models, aiming to enhance robustness and adaptability in complex traffic scenarios.

\appendices
\section{Proof of Theorem~\ref{thm:thm1}}\label{app:proof1}
\renewcommand\theequation{A.\arabic{equation}}
\setcounter{equation}{0}
We provide the proof of Theorem~\ref{thm:thm1} by constructing a quadratic Lyapunov functional of the form:
\begin{align}
    V(\tilde{x}(t),Z(t)) = \tilde{x}^\top(t)  P(Z(t)) \tilde{x}(t), \quad P(Z(t)) > 0.
\end{align}
where $0<P(Z(t)) \in \mathbb{R}^{n \times n}$, $Z(t) \in \mathcal{V}$. For simplicity, we denote $P(Z(t))=P_{ik}$ and $V(\tilde{x}(t),Z(t))=V(\tilde{x}(t),(i,k))$ when $Z(t)=(\eta(t), \hat{\eta}(t)) = (i, k) \in \mathcal{V}$.

Given $Z(t)=(i, k)$ and a small positive scalar $\Delta t$, an operator $\mathcal{L}$ is defined to represent the weak infinitesimal generator:
\begin{align}
    \mathcal{L}V(\tilde{x}(t),Z(t)=(i, k))=\lim_{\Delta t \to 0}\frac{1}{\Delta t}[\mathbb{E}\{V(\tilde{x}(t+\Delta t),Z(t+\Delta t))\} - V(\tilde{x}(t),(i, k)) ].
\end{align}
We have
\begin{align}
    \tilde{x}(t+\Delta t) = [I+A_{ik}\Delta t]\tilde{x}(t)+D_{ik}\Delta t \tilde{v}_\leadv(t)+\tilde{o}(\Delta t)
\end{align}
where $\tilde{o}(\Delta t) \in \mathbb{R}^n$ and satisfies $\lim_{\Delta t \to 0} (\tilde{o}(\Delta t)/{\Delta t})=0_{n \times 1}$; $A_{ik} = A+B(K_{\human,i}+K_{\av,k})$, $D_{ik} = D+B(D_{\human,i}+D_{\av,k})$. We obtain:
\begin{align}
    &\mathcal{L}{V}(\tilde{x}(t), (i,k)) \nonumber\\
    =& \lim_{\Delta t \to 0}\frac{1}{\Delta t}[\mathbb{E}\{ V(\tilde{x}(t+\Delta t), (j,l))\} - V(\tilde{x}(t), (i,k))]
    \nonumber \\
    =& \lim_{\Delta t \to 0}\frac{1}{\Delta t}[\mathbb{E}\{ \sum_{(j,l) \neq (i,k)} \nu_{(i,k)(j,l)} \Delta t (\tilde{x}^\top (t)(I+A_{ik}\Delta t)^\top  +\tilde{v}_{\leadv}^\top (t)D_{ik}^\top \Delta t )P_{jl}( (I+A_{ik}\Delta t)\tilde{x}(t)
    + D_{ik}\Delta t \tilde{v}_\leadv(t)) \nonumber\\
    &+(1+\nu_{(i,k)(i,k)}\Delta t)(\tilde{x}^\top (t)(I+A_{ik}\Delta t)^\top + \tilde{v}_{\leadv}^\top (t)D_{ik}^\top \Delta t)P_{ik}((I+A_{ik}\Delta t)\tilde{x}(t)+ D_{ik}\Delta t \tilde{v}_\leadv(t)) \}  \nonumber \\
    &- \tilde{x}^\top (t)P_{ik}\tilde{x}(t) + o(\Delta t)]\nonumber\\
    =& \tilde{x}^\top(t)  \left(A_{ik}^\top  P_{ik} + P_{ik} A_{ik} \right) \tilde{x}(t) + \tilde{x}^\top(t)  P_{ik}D_{ik}\tilde{v}_\leadv(t) + \tilde{v}_\leadv^\top(t) D_{ik}^\top P_{ik}\tilde{x}(t)+ \tilde{x}^\top(t) ( \sum_{(j,l)\in \mathcal{V}} \nu_{(i,k)(j,l)} P_{jl} ) \tilde{x}(t).
\end{align}

Suppose that there are matrices $X_{ik}>0$, $G_k$ and $V_k$, and scalar $\varepsilon>0$ with $i\in \mathcal{N}$ and $k\in \mathcal{M}$ such that \eqref{eq:LMI0} to \eqref{eq:LMI_con5} are satisfied. We know that $G_k$ is non-singular. By setting $K_{\av,k} = V_kG_k^{-1}$ and $D_{\av,k} = L_k$, we get that
\begin{align}
    \Omega_{ik} &= (A+BK_{\human,i})G_k + BV_k \nonumber\\
                &= (A+BK_{\human,i}+BK_{\av,k})G_k \\
    \Phi_{ik} &= D+BD_{\human,i}+BD_{\av,k}
\end{align}
Thus we conclude that \eqref{eq:LMI0} is equivalent to 
\begin{align}\label{eq:proof1}
    \begin{bmatrix}
        \nu_{(i,k)(i,k)}X_{ik} & \Phi_{ik} & 0 &X_{ik} & X_{ik}\Pi_{ik} \\
        \Phi_{ik}^\top & -\gamma^2I & 0 &0 & 0\\
        0 & 0 & -I & 0 & 0\\
        X_{ik}^\top & 0 &0 & 0 & 0\\
        \Pi_{ik}^\top X_{ik}^\top &0 & 0 & 0 & -\mathcal{D}_{ik}
    \end{bmatrix} + \mathrm{Her}\left(\begin{bmatrix}
         (A+BK_{\human,i}+BK_{\av,k})\\
        0 \\
        C \\
        -I \\
        0
    \end{bmatrix} G_k\begin{bmatrix}
        \varepsilon I\\
        0\\
        0\\
        I\\
        0
    \end{bmatrix}^\top \right)<0
\end{align}
We recall that $A_{ik} = A+B(K_{\human,i}+K_{\av,k})$, $D_{ik} = D+B(D_{\human,i}+D_{\av,k})$ and notice that
\begin{align}
    \begin{bmatrix}
        A_{ik}^\top & 0 & C^\top & -I & 0
    \end{bmatrix}W = 0
\end{align}
where $W$ is defined as
\begin{align}
    W = \begin{bmatrix}
        I & 0 & 0 & 0\\
        0 & I & 0 & 0\\
        0 & 0 & I & 0\\
        A_{ik}^\top & 0 & C^\top & 0 \\
        0 & 0 & 0 & I
    \end{bmatrix}
\end{align}
We have from the Finsler's Lemma that \eqref{eq:proof1} is equivalent to
\begin{align}
    W^\top \begin{bmatrix}
        \nu_{(i,k)(i,k)}X_{ik} & \Phi_{ik} & 0 &X_{ik} & X_{ik}\Pi_{ik} \\
        \Phi_{ik}^\top & -\gamma^2I & 0 &0 & 0\\
        0 & 0 & -I & 0 & 0\\
        X_{ik}^\top & 0 &0 & 0 & 0\\
        \Pi_{ik}^\top X_{ik}^\top &0 & 0 & 0 & -\mathcal{D}_{ik}
    \end{bmatrix} W<0
\end{align}
resulting in
\begin{align}
    \begin{bmatrix}
        \nu_{(i,k)(i,k)}X_{ik}+\mathrm{Her}(X_{ik}A_{ik}^\top)& \Phi_{ik}  &X_{ik}C^\top & X_{ik}\Pi_{ik} \\
        \Phi_{ik}^\top & -\gamma^2I & 0  & 0\\
        CX_{ik}^\top & 0  & -I & 0\\
        \Pi_{ik}^\top X_{ik}^\top &0 & 0  & -\mathcal{D}_{ik}
    \end{bmatrix}<0.
\end{align}
By applying Schur Complement, we get
\begin{align}\label{eq:proof2}
    \begin{bmatrix}
        \Gamma_{ik} & \Phi_{ik} & X_{ik} C^\top\\
        \Phi_{ik}^\top & -\gamma^2 I & 0 \\
        C X_{ik}^\top & 0 & -I
    \end{bmatrix}<0
\end{align}
with 
\begin{align}
    \Gamma_{ik} = \nu_{(i,k)(i,k)}X_{ik}&+X_{ik}A_{ik}^\top+A_{ik} X_{ik}^\top +X_{ik}\Pi_{ik}\mathcal{D}_{ik}^{-1}\Pi_{ik}^\top X_{ik}^\top.
\end{align}
Since we define $\mathcal{V}_{(i,k)}$, $\Pi_{ik}$ and $\mathcal{D}_{ik}$ as follows:
\begin{align}
    \begin{cases}
        \mathcal{V}_{(i,k)} &= \{ (j,l)\in \mathcal{V}; (j,l) \ne (i,k)\text{ and } \nu_{(i,k)(j,l)} \ne 0\} \nonumber\\
     &= \{r_{(i,k)}(1),\ldots,r_{(i,k)}(\tau_{(i,k)}); r_{(i,k)}(\iota)\in \mathcal{V},\iota  = 1,\ldots, \tau_{(i,k)}\} \\
    \Pi_{ik} =& \begin{bmatrix}
        \sqrt{\nu_{(i,k)r_{(i,k)}(1)}}I & \ldots & \sqrt{\nu_{(i,k)r_{(i,k)}(\tau_{(i,k)})}}I
    \end{bmatrix}\\
    \mathcal{D}_{ik}=& \mathrm{diag}(X_{r_{(i,k)}(1)},\ldots,X_{r_{(i,k)}(\tau_{(i,k)})})
    \end{cases}
\end{align}
Let $P_{ik} = X_{ik}^{-1}$, so we obtain: 
\begin{align}
    \Pi_{ik}\mathcal{D}_{ik}^{-1}\Pi_{ik} &= \sum_{(j,l)\in \mathcal{V}_{(i,k)}} \nu_{(i,k)(j,l)}X_{jl}^{-1} \nonumber\\
    & = \sum_{(j,l)\in \mathcal{V}} \nu_{(i,k)(j,l)} X_{jl}^{-1} - \nu_{(i,k)(i,k)}X_{ik}^{-1} \nonumber\\
    & = \sum_{(j,l)\in \mathcal{V}} \nu_{(i,k)(j,l)} P_{jl}- \nu_{(i,k)(i,k)}P_{ik}
\end{align}
Now by pre- and post-multiplying \eqref{eq:proof2} by $\mathrm{diag}(X_{ik}^{-1},I,I)$ and substituting $X_{ik}^{-1}$ by $P_{ik}$ we get that
\begin{align}\label{eq:proof3}
&\begin{bmatrix}
A_{ik}^\top  P_{ik} + P_{ik} A_{ik} + \sum_{(j,l)} \nu_{(i,k)(j,l)} P_{jl}& P_{ik} D_{ik} & C^\top\\
D_{ik}^\top  P_{ik} & -\gamma^2 I &0 \\
C & 0 & -I
\end{bmatrix} < 0,
\end{align}

Then we need to prove that \eqref{eq:proof3} is sufficient to stochastic $\mathcal{L}_2$ string stability. 
We conclude from \eqref{eq:proof3} that:
\begin{align}
   & \tilde{x}^\top(t)  (A_{ik}^\top  P_{ik} + P_{ik} A_{ik} +  \sum_{(j,l)\in \mathcal{V}} \nu_{(i,k)(j,l)} P_{jl} +C^\top C ) \tilde{x}(t) \nonumber \\
   &+ \tilde{x}^\top (t)P_{ik}D_{ik}\tilde{v}_\leadv(t) + \tilde{v}_\leadv^\top(t) D_{ik}^\top P_{ik}\tilde{x}(t) -\gamma^2\tilde{v}_\leadv^\top(t)  \tilde{v}_\leadv(t)<0,
\end{align}
so we have:
\begin{align}\label{eq:lya}
&\mathcal{L}{V}(\tilde{x}(t), (i,k)) + \mathbb{E} \left\{ \tilde{x}^\top (t) C^\top  C \tilde{x}(t) \right\} - \gamma^2 \mathbb{E} \left\{ \tilde{v}_\leadv^\top(t)  \tilde{v}_\leadv(t) \right\} < 0.
\end{align} 
Integrating \eqref{eq:lya} from 0 to $t$ yields
\begin{align}
    \mathbb{E}[V(\tilde{x}(t),(i,k))]&-V(\tilde{x}(0),Z(0)) + \mathbb{E} \left[ \int_0^t \tilde{z}^\top(s) \tilde{z}(s) ds \right] \leq \gamma^2 \mathbb{E} \left[ \int_0^t \tilde{v}_\leadv^\top(s) \tilde{v}_\leadv(s) ds \right]
\end{align}
Under the initial state $\tilde{x}(0)= 0$, we have $V (\tilde{x}(0), Z(0))= 0$. Since $\mathbb{E}\{V (\tilde{x}(t), (i,k) )\} > 0$ for any $\tilde{x}(t) \neq 0$, we have
\begin{align}
\mathbb{E} \left\{ \int_{0}^{\infty} \tilde{z}^\top(s)  \tilde{z}(s) ds \right\} \leq \gamma^2 \mathbb{E} \left\{ \int_{0}^{\infty} \tilde{v}_\leadv^\top(s)  \tilde{v}_\leadv(s) ds \right\}, \gamma \leq 1.
\end{align}
this proves Lemma \ref{lem:L2}. The proof is complete.
\begin{remark}
Although Theorem~\ref{thm:thm1} guarantees stochastic $\mathcal{L}_2$ string stability whenever the LMIs are feasible, the synthesized gains may occasionally become excessively large due to numerical ill-conditioning. To address this, we impose an additional structural constraint $V_kR=0$, where $R$ spans the dominant directions extracted from $K_{\human,i}$. This regularization keeps the controller well-conditioned and prevents unrealistic intervention spikes, while maintaining the same stability guarantee.
\end{remark}

\section{Proof of Theorem~\ref{thm:mic_lmi}}\label{app:proof2}
\renewcommand\theequation{B.\arabic{equation}}
\setcounter{equation}{0}
We show that condition \eqref{eq:LMI1} is sufficient by proving it implies the dissipation inequality for the augmented output, and thus guarantees the performance bound in Proposition~\ref{prop:mic_perf}. As in Theorem~\ref{thm:thm1}, we consider the quadratic Lyapunov functional $V(\tilde{x}(t),Z(t))=\tilde{x}^\top(t)P(Z(t))\tilde{x}(t)$ with $P_{ik}=P(Z(t)=(i,k))>0$, and define $A_{ik}=A+B(K_{\human,i}+K_{\av,k})$ and $D_{ik}=D+B(D_{\human,i}+D_{\av,k})$. The generator evaluation is identical to Theorem \ref{thm:thm1} and is omitted, so we focus on the additional terms introduced by the augmented output \eqref{eq:aug_output}.

Compared with Theorem~\ref{thm:thm1}, the augmented output contributes the design weight $\beta$ through the blocks $\beta V_k$ and $\beta L_k$. The Markov jump coupling remains encoded by $\Pi_{ik}$ and $\mathcal D_{ik}$. We introduce matrices $X_{ik}>0$, $G_k$ and $V_k$ as in Theorem~\ref{thm:thm1}, and set $K_{\av,k}=V_k G_k^{-1}$ and $D_{\av,k}=L_k$. Substituting these into $A_{ik}$ and $D_{ik}$, and applying the same slack-variable lifting as in Theorem~\ref{thm:thm1}, and using Finsler's lemma to eliminate the explicit dependence on $G_k$ convert \eqref{eq:LMI1} into the following equivalent block inequality:
\begin{align}
\begin{bmatrix}
\nu_{(i,k)(i,k)}X_{ik}+\mathrm{Her}(X_{ik}A_{ik}^\top) & \Phi_{ik} & X_{ik}C^\top & \beta X_{ik}K_{\av,k}^\top & X_{ik}\Pi_{ik} \\
\Phi_{ik}^\top & -\gamma^2 I & 0 & \beta D_{\av,k}^\top & 0 \\
C X_{ik}^\top & 0 & -I & 0 & 0\\
\beta K_{\av,k}X_{ik}^\top & \beta D_{\av,k} & 0 & -I & 0\\
\Pi_{ik}^\top X_{ik}^\top & 0 & 0 & 0 & -\mathcal D_{ik}
\end{bmatrix}<0.
\end{align}
where $\Phi_{ik}$, $\Pi_{ik}$, and $\mathcal{D}_{ik}$ are consistent with \eqref{eq:LMI_con2} to \eqref{eq:LMI_con5}. 

Applying the Schur complement with respect to the negative definite sub-blocks and substituting $P_{ik}=X_{ik}^{-1}$ yield \eqref{eq:mic_preconvex}. The terms scaled by $\beta$ in the lifted inequality become quadratic in $\beta$ after this reduction, which reflects the fact that the automated-effort channel enters through cross products. The Markov jump coupling carried by $\Pi_{ik}$ and $\mathcal D_{ik}$ reduces to $\sum_{(j,l)}\nu_{(i,k)(j,l)}P_{jl}$ as in Theorem~\ref{thm:thm1}.
\begin{align}
\begin{bmatrix}
A_{ik}^\top  P_{ik} + P_{ik} A_{ik} + \sum_{(j,l)} \nu_{(i,k)(j,l)} P_{jl}+\beta^2 K_{\av,k}^\top K_{\av,k}& P_{ik} D_{ik} +\beta^2 K_{\av,k}^\top D_{\av,k}& C^\top\\
D_{ik}^\top P_{ik} +\beta^2 D_{\av,k}^\top K_{\av,k}& \beta^2 D_{\av,k}^\top D_{\av,k}-\gamma^2 I &0 \\
C & 0 & -I
\end{bmatrix} < 0.
\label{eq:mic_preconvex}
\end{align}

With the augmented output definition \eqref{eq:aug_output} and $u_{\av}=K_{\av,k}\tilde{x}+D_{\av,k}\tilde{v}_\leadv$, the term $\tilde{z}^\top\tilde{z}$ expands as
\begin{align}
\tilde{z}^\top \tilde{z}
= \tilde{x}^\top C^\top C \tilde{x} 
+ \beta^2 \tilde{x}^\top K_{\av,k}^\top K_{\av,k}\tilde{x}
+ \beta^2 \tilde{x}^\top K_{\av,k}^\top D_{\av,k}\tilde{v}_\leadv
+ \beta^2 \tilde{v}_\leadv^\top D_{\av,k}^\top K_{\av,k}\tilde{x}
+ \beta^2 \tilde{v}_\leadv^\top D_{\av,k}^\top D_{\av,k}\tilde{v}_\leadv.
\end{align}

Hence, the matrix inequality \eqref{eq:mic_preconvex} is equivalent to the dissipation condition
\begin{align}
\mathcal{L}V(\tilde{x}(t),(i,k)) + \tilde{z}^\top(t)\tilde{z}(t) - \gamma^2 \tilde{v}_\leadv^\top(t)\tilde{v}_\leadv(t) < 0, \label{eq:mic_diss}
\end{align}

Integrating \eqref{eq:mic_diss} over $[0,t]$, taking expectations with respect to the jump process, and letting $t\to\infty$ with $\tilde{x}(0)=0$ gives
\begin{align*}
\mathbb{E}\left\{\int_0^\infty \tilde{z}^\top(s)\tilde{z}(s)\,ds\right\}
\le \gamma^2 \mathbb{E}\!\left\{\int_0^\infty \tilde{v}_\leadv^\top(s)\tilde{v}_\leadv(s)\,ds\right\},
\end{align*}
which coincides with the performance condition \eqref{eq:H_infity_energy_mic}. This completes the proof.
\begin{remark}
Similar to the nominal controller case in Theorem~\ref{thm:thm1}, solving the LMIs in Theorem~\ref{thm:mic_lmi} may lead to numerically ill-conditioned gains. To improve robustness, the same regularization $V_kR=0$ is imposed, which ensures well-conditioned synthesis without affecting the theoretical guarantees.
\end{remark}

\section{Sensitivity Analysis of Estimation Parameters}\label{app:sensitivity}
\renewcommand\theequation{C.\arabic{equation}}
\setcounter{equation}{0}

This appendix examines the reliability of the proposed controller under errors in mode observation. The observation chain $\hat{\eta}(t)$ is governed by the parameters $\alpha$ and $q$. We vary them over $\alpha\in[0.05,0.30]$ and $q\in[0.02,0.30]~\text{s}^{-1}$, which cover moderate to severe levels of estimation noise and delay.

For each $(\alpha,q)$ pair, Monte Carlo simulations are conducted and three metrics are recorded: (i) the observation accuracy, defined as the fraction of time with $\hat{\eta}(t)=\eta(t)$; (ii) the average follow delay, i.e., the lag for $\hat{\eta}(t)$ to catch up after true mode switches (capped at 0.5 s); and (iii) the $\mathcal{L}_2$ gain $\gamma$, which certifies stochastic string stability of the closed-loop system.

\begin{figure}[t]
    \centering
    \subfigure[Accuracy]{
        \includegraphics[width=0.31\linewidth]{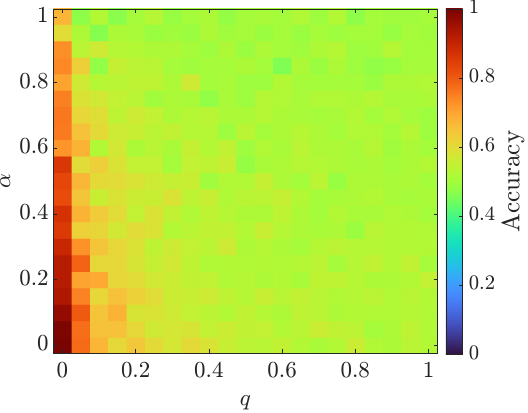}\label{app_fig:accuracy}
    }
    \subfigure[Delay]{
        \includegraphics[width=0.31\linewidth]{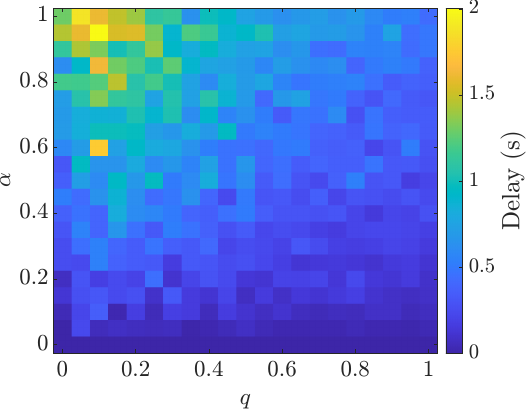}\label{app_fig:delay}
    }
    \subfigure[Stability]{
        \includegraphics[width=0.31\linewidth]{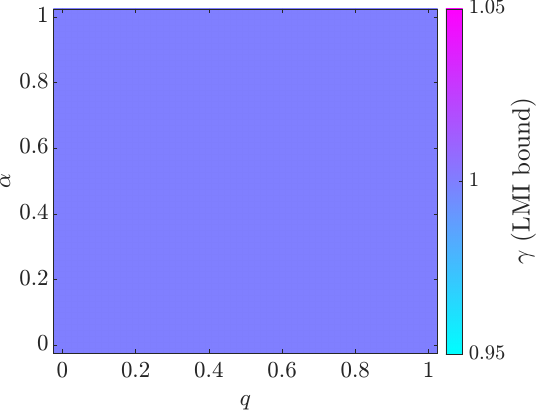}\label{app_fig:gamma}
    }
    \caption{Sensitivity analysis results.}
    \label{fig:sensitivity}
\end{figure}

The sensitivity results are shown in Fig.~\ref{fig:sensitivity}.
Fig.~\ref{app_fig:accuracy} indicates that observation accuracy is especially sensitive to $\alpha$: as $\alpha$ increases, accuracy drops significantly, whereas the influence of $q$ is comparatively minor.
Fig.~\ref{app_fig:delay} shows that the follow delay remains low when both $\alpha$ and $q$ are small, but increases with larger values; the delay is more sensitive to $\alpha$ than to $q$.
Fig.~\ref{app_fig:gamma} demonstrates that within the tested ranges, the LMIs are feasible and the synthesized controllers consistently yield $\mathcal{L}_2$ gains $\gamma$ close to $1$, ensuring string stability. In particular, when $\alpha=0$ and $q=0$, the observation becomes perfect with accuracy approaching $1$, delay nearly $0$, and  $\gamma$ approximately equal to $1$.

In summary, the controller remains feasible and stabilizing across the tested ranges of $\alpha$ and $q$, confirming reliable $\mathcal{L}_2$ string stability even under moderate observation errors.

\bibliographystyle{plain}
\bibliography{references.bib}

\begin{thebibliography}{10}

\bibitem{abbink2012haptic}
David~A Abbink, Mark Mulder, and Erwin~R Boer.
\newblock Haptic shared control: smoothly shifting control authority?
\newblock {\em Cognition, Technology \& Work}, 14(1):19--28, 2012.

\bibitem{ali2021clacd}
Yasir Ali, Zuduo Zheng, Md~Mazharul Haque, Mehmet Yildirimoglu, and Simon Washington.
\newblock Clacd: A complete lane-changing decision modeling framework for the connected and traditional environments.
\newblock {\em Transportation Research Part C: Emerging Technologies}, 128:103162, 2021.

\bibitem{ammourah2025introduction}
Rami Ammourah, Pedram Beigi, Bingyi Fan, Samer~H Hamdar, John Hourdos, Chun-Chien Hsiao, Rachel James, Mohammdreza Khajeh-Hosseini, Hani~S Mahmassani, Dana Monzer, et~al.
\newblock Introduction to the third generation simulation dataset: Data collection and trajectory extraction.
\newblock {\em Transportation Research Record}, 2679(1):1768--1784, 2025.

\bibitem{bevly2016lane}
David Bevly, Xiaolong Cao, Mikhail Gordon, Guchan Ozbilgin, David Kari, Brently Nelson, Jonathan Woodruff, Matthew Barth, Chase Murray, Arda Kurt, et~al.
\newblock Lane change and merge maneuvers for connected and automated vehicles: A survey.
\newblock {\em IEEE Transactions on Intelligent Vehicles}, 1(1):105--120, 2016.

\bibitem{bhattacharyya2023automated}
Viranjan Bhattacharyya and Ardalan Vahidi.
\newblock Automated vehicle highway merging: Motion planning via adaptive interactive mixed-integer mpc.
\newblock In {\em 2023 American Control Conference (ACC)}, pages 1141--1146, 2023.

\bibitem{Bhattacharyya2025LaneChange}
Viranjan Bhattacharyya and Ardalan Vahidi.
\newblock Automated lane change via adaptive interactive mpc: Human-in-the-loop experiments.
\newblock {\em IEEE Transactions on Control Systems Technology}, 33(4):1246--1257, 2025.

\bibitem{bian2020lanekeeping}
Yougang Bian, Jieyun Ding, Manjiang Hu, Qing Xu, Jianqiang Wang, and Keqiang Li.
\newblock An advanced lane-keeping assistance system with switchable assistance modes.
\newblock {\em IEEE Transactions on Intelligent Transportation Systems}, 21(1):385--396, 2020.

\bibitem{chen2025Collaborative}
Xiang Chen, Wenlong Ding, Wanzhong Zhao, and Chunyan Wang.
\newblock Collaborative control of human–machine game in lateral and longitudinal dimensions considering dynamic allocation of driving authority.
\newblock {\em IEEE Transactions on Systems, Man, and Cybernetics: Systems}, 55(6):4309--4321, 2025.

\bibitem{chen2025dynamic}
Yutao Chen, Hongliang Zhang, Haocong Chen, Jie Huang, Bin Wang, Zixiang Xiong, Yuyi Wang, and Xiwen Yuan.
\newblock Dynamic control authority allocation in indirect shared control for steering assistance.
\newblock {\em IEEE Transactions on Intelligent Transportation Systems}, 26(3):3458--3470, 2025.

\bibitem{Dragan2013APF}
Anca~D. Dragan and Siddhartha~S. Srinivasa.
\newblock A policy-blending formalism for shared control.
\newblock {\em The International Journal of Robotics Research}, 32:790 -- 805, 2013.

\bibitem{DU2019428}
Na~Du, Jacob Haspiel, Qiaoning Zhang, Dawn Tilbury, Anuj~K. Pradhan, X.~Jessie Yang, and Lionel~P. Robert.
\newblock Look who’s talking now: Implications of av’s explanations on driver’s trust, av preference, anxiety and mental workload.
\newblock {\em Transportation Research Part C: Emerging Technologies}, 104:428--442, 2019.

\bibitem{erlien2016sharedsteering}
Stephen~M. Erlien, Susumu Fujita, and Joseph~Christian Gerdes.
\newblock Shared steering control using safe envelopes for obstacle avoidance and vehicle stability.
\newblock {\em IEEE Transactions on Intelligent Transportation Systems}, 17(2):441--451, 2016.

\bibitem{fang2023humanmachine}
Zhenwu Fang, Jinxiang Wang, Zejiang Wang, Jinhao Liang, Yahui Liu, and Guodong Yin.
\newblock A human-machine shared control framework considering time-varying driver characteristics.
\newblock {\em IEEE Transactions on Intelligent Vehicles}, 8(7):3826--3838, 2023.

\bibitem{GAO2024104491}
Jianqiang Gao, Bo~Yu, Yuren Chen, Shan Bao, Kun Gao, and Lanfang Zhang.
\newblock An adas with better driver satisfaction under rear-end near-crash scenarios: A spatio-temporal graph transformer-based prediction framework of evasive behavior and collision risk.
\newblock {\em Transportation Research Part C: Emerging Technologies}, 159:104491, 2024.

\bibitem{guo2015sharedmerge}
Chunshi Guo, Chouki Sentouh, Jean-Christophe Popieul, Boussaad Soualmi, and Jean-Baptiste Haué.
\newblock Shared control framework applied for vehicle longitudinal control in highway merging scenarios.
\newblock In {\em 2015 IEEE International Conference on Systems, Man, and Cybernetics}, pages 3098--3103, 2015.

\bibitem{guo2024game}
Hongyan Guo, Wanqing Shi, Jun Liu, Jingzheng Guo, Qingyu Meng, Dongpu Cao, and Hong Chen.
\newblock Game-theoretic human-machine shared steering control strategy under extreme conditions.
\newblock {\em IEEE Transactions on Intelligent Vehicles}, 9(1):2766--2779, 2024.

\bibitem{han2022adaptivesteering}
Jiayi Han, Jian Zhao, Bing Zhu, and Dongjian Song.
\newblock Adaptive steering torque coupling framework considering conflict resolution for human-machine shared driving.
\newblock {\em IEEE Transactions on Intelligent Transportation Systems}, 23(8):10983--10995, 2022.

\bibitem{hu2019trajectory}
Xiangwang Hu and Jian Sun.
\newblock Trajectory optimization of connected and autonomous vehicles at a multilane freeway merging area.
\newblock {\em Transportation Research Part C: Emerging Technologies}, 101:111--125, 2019.

\bibitem{huang2024safetyaware}
Wenhui Huang, Haochen Liu, Zhiyu Huang, and Chen Lv.
\newblock Safety-aware human-in-the-loop reinforcement learning with shared control for autonomous driving.
\newblock {\em IEEE Transactions on Intelligent Transportation Systems}, 25(11):16181--16192, 2024.

\bibitem{HUANG2025105262}
Zilin Huang, Zihao Sheng, and Sikai Chen.
\newblock Pe-rlhf: Reinforcement learning with human feedback and physics knowledge for safe and trustworthy autonomous driving.
\newblock {\em Transportation Research Part C: Emerging Technologies}, 179:105262, 2025.

\bibitem{JIANG2020safeteleoperation}
Frank~J. Jiang, Yulong Gao, Lihua Xie, and Karl~H. Johansson.
\newblock Human-centered design for safe teleoperation of connected vehicles.
\newblock {\em IFAC-PapersOnLine}, 53(5):224--231, 2020.

\bibitem{kim2012switchingcons}
Dae-Jin Kim, Rebekah Hazlett-Knudsen, Heather Culver-Godfrey, Greta Rucks, Tara Cunningham, David Portee, John Bricout, Zhao Wang, and Aman Behal.
\newblock How autonomy impacts performance and satisfaction: Results from a study with spinal cord injured subjects using an assistive robot.
\newblock {\em IEEE Transactions on Systems, Man, and Cybernetics - Part A: Systems and Humans}, 42(1):2--14, 2012.

\bibitem{lang2023shared}
Yilin Lang, Zhaoyang Li, Jiazhe Guo, Wenxin Zhu, and Qinyuan Ren.
\newblock A shared control approach for autonomous vehicles via driver behaviors learning.
\newblock In {\em 2023 IEEE 12th Data Driven Control and Learning Systems Conference (DDCLS)}, pages 916--921, 2023.

\bibitem{li2017estimation}
Guofa Li, Shengbo~Eben Li, Bo~Cheng, and Paul Green.
\newblock Estimation of driving style in naturalistic highway traffic using maneuver transition probabilities.
\newblock {\em Transportation Research Part C: Emerging Technologies}, 74:113--125, 2017.

\bibitem{Li2023DeepReinforcementLearning}
Guofa Li, Yifan Qiu, Yifan Yang, Zhenning Li, Shen Li, Wenbo Chu, Paul Green, and Shengbo~Eben Li.
\newblock Lane change strategies for autonomous vehicles: A deep reinforcement learning approach based on transformer.
\newblock {\em IEEE Transactions on Intelligent Vehicles}, 8(3):2197--2211, 2023.

\bibitem{li2021indirect}
Renjie Li, Yanan Li, Shengbo~Eben Li, Chaofei Zhang, Etienne Burdet, and Bo~Cheng.
\newblock Indirect shared control for cooperative driving between driver and automation in steer-by-wire vehicles.
\newblock {\em IEEE Transactions on Intelligent Transportation Systems}, 22(12):7826--7836, 2021.

\bibitem{li2022adaptive}
Xueyun Li, Yiping Wang, Chuqi Su, Xinle Gong, Jin Huang, and Dengke Yang.
\newblock Adaptive authority allocation approach for shared steering control system.
\newblock {\em IEEE Transactions on Intelligent Transportation Systems}, 23(10):19428--19439, 2022.

\bibitem{LI2020switchingcontrol}
Yang Li, Dihua Sun, Min Zhao, Jin Chen, Zhongcheng Liu, Senlin Cheng, and Tao Chen.
\newblock Mpc-based switched driving model for human vehicle co-piloting considering human factors.
\newblock {\em Transportation Research Part C: Emerging Technologies}, 115:102612, 2020.

\bibitem{li2020automatic}
Zhaolun Li, Jingjing Jiang, and Wen-Hua Chen.
\newblock Automatic lane change maneuver in dynamic environment using model predictive control method.
\newblock In {\em 2020 IEEE/RSJ International Conference on Intelligent Robots and Systems (IROS)}, pages 2384--2389, 2020.

\bibitem{Li2019stringstability}
Zhicheng Li, Bin Hu, Ming Li, and Gengnan Luo.
\newblock String stability analysis for vehicle platooning under unreliable communication links with event-triggered strategy.
\newblock {\em IEEE Transactions on Vehicular Technology}, 68(3):2152--2164, 2019.

\bibitem{liu2018dynamic}
Kai Liu, Jianwei Gong, Arda Kurt, Huiyan Chen, and Umit Ozguner.
\newblock Dynamic modeling and control of high-speed automated vehicles for lane change maneuver.
\newblock {\em IEEE Transactions on Intelligent Vehicles}, 3(3):329--339, 2018.

\bibitem{liu2023Shared}
Rui Liu, Xuan Zhao, Xichan Zhu, and Jian Ma.
\newblock A human-like shared driving strategy in lane-changing scenario using cooperative lpv/mpc.
\newblock {\em IEEE Transactions on Intelligent Transportation Systems}, 24(9):9915--9928, 2023.

\bibitem{LIU2023stabilitysemimarkov}
Yang-Fan Liu, Huai-Ning Wu, and Xiu-Mei Zhang.
\newblock Stability and h-ifinity performance of human-in-the-loop control systems through hidden semi-markov human behavior modeling.
\newblock {\em Applied Mathematical Modelling}, 116:799--815, 2023.

\bibitem{luo2016dynamic}
Yugong Luo, Yong Xiang, Kun Cao, and Keqiang Li.
\newblock A dynamic automated lane change maneuver based on vehicle-to-vehicle communication.
\newblock {\em Transportation Research Part C: Emerging Technologies}, 62:87--102, 2016.

\bibitem{MONTEIRO2023104138}
Fernando~V. Monteiro and Petros Ioannou.
\newblock Safe autonomous lane changes and impact on traffic flow in a connected vehicle environment.
\newblock {\em Transportation Research Part C: Emerging Technologies}, 151:104138, 2023.

\bibitem{nilsson2017lane}
Julia Nilsson, Mattias Brännström, Erik Coelingh, and Jonas Fredriksson.
\newblock Lane change maneuvers for automated vehicles.
\newblock {\em IEEE Transactions on Intelligent Transportation Systems}, 18(5):1087--1096, 2017.

\bibitem{PAN2016403}
T.L. Pan, William~H.K. Lam, A.~Sumalee, and R.X. Zhong.
\newblock Modeling the impacts of mandatory and discretionary lane-changing maneuvers.
\newblock {\em Transportation Research Part C: Emerging Technologies}, 68:403--424, 2016.

\bibitem{saifuzzaman2015revisiting}
Mohammad Saifuzzaman, Zuduo Zheng, Md~Mazharul Haque, and Simon Washington.
\newblock Revisiting the task--capability interface model for incorporating human factors into car-following models.
\newblock {\em Transportation research part B: methodological}, 82:1--19, 2015.

\bibitem{salvato2024stopandgo}
Erica Salvato, Lorenzo Elia, Gianfranco Fenu, and Thomas Parisini.
\newblock Stop-and-go traffic wave attenuation: A shared control approach.
\newblock In {\em 2024 IEEE 63rd Conference on Decision and Control (CDC)}, pages 4929--4934, 2024.

\bibitem{Saxena2020DrivinginDense}
Dhruv~Mauria Saxena, Sangjae Bae, Alireza Nakhaei, Kikuo Fujimura, and Maxim Likhachev.
\newblock Driving in dense traffic with model-free reinforcement learning.
\newblock In {\em 2020 IEEE International Conference on Robotics and Automation (ICRA)}, pages 5385--5392, 2020.

\bibitem{Sentouh2019driverautomation}
Chouki Sentouh, Anh-Tu Nguyen, Mohamed~Amir Benloucif, and Jean-Christophe Popieul.
\newblock Driver-automation cooperation oriented approach for shared control of lane keeping assist systems.
\newblock {\em IEEE Transactions on Control Systems Technology}, 27(5):1962--1978, 2019.

\bibitem{shi2022human}
Zhuqing Shi, Hong Chen, Ting Qu, and Shuyou Yu.
\newblock Human–machine cooperative steering control considering mitigating human–machine conflict based on driver trust.
\newblock {\em IEEE Transactions on Human-Machine Systems}, 52(5):1036--1048, 2022.

\bibitem{stadtmann2017h_2}
F.~Stadtmann and O.~L.~V. Costa.
\newblock $h_2$-control of continuous-time hidden markov jump linear systems.
\newblock {\em IEEE Transactions on Automatic Control}, 62(8):4031--4037, 2017.

\bibitem{Tran2019mpc}
Anh~Tuan Tran, Masato Kawaguchi, Hiroyuki Okuda, and Tatsuya Suzuki.
\newblock A model predictive control-based lane merging strategy for autonomous vehicles.
\newblock In {\em 2019 IEEE Intelligent Vehicles Symposium (IV)}, pages 594--599, 2019.

\bibitem{treiber2013microscopic}
Martin Treiber and Arne Kesting.
\newblock Microscopic calibration and validation of car-following models--a systematic approach.
\newblock {\em Procedia-Social and Behavioral Sciences}, 80:922--939, 2013.

\bibitem{Triest2020Learning}
Samuel Triest, Adam Villaflor, and John~M. Dolan.
\newblock Learning highway ramp merging via reinforcement learning with temporally-extended actions.
\newblock In {\em 2020 IEEE Intelligent Vehicles Symposium (IV)}, pages 1595--1600, 2020.

\bibitem{tsoi2010lanekeeping}
Kakin~K. Tsoi, Mark Mulder, and David~A. Abbink.
\newblock Balancing safety and support: Changing lanes with a haptic lane-keeping support system.
\newblock In {\em 2010 IEEE International Conference on Systems, Man and Cybernetics}, pages 1236--1243, 2010.

\bibitem{WANG201573}
Meng Wang, Serge~P. Hoogendoorn, Winnie Daamen, Bart {van Arem}, and Riender Happee.
\newblock Game theoretic approach for predictive lane-changing and car-following control.
\newblock {\em Transportation Research Part C: Emerging Technologies}, 58:73--92, 2015.

\bibitem{Wang2023Interaction}
Renzi Wang, Mathijs Schuurmans, and Panagiotis Patrinos.
\newblock Interaction-aware model predictive control for autonomous driving.
\newblock In {\em 2023 European Control Conference (ECC)}, pages 1--6, 2023.

\bibitem{wang2017drivingstyleanalysis}
Wenshuo Wang, Junqiang Xi, and Ding Zhao.
\newblock Driving style analysis using primitive driving patterns with bayesian nonparametric approaches.
\newblock {\em IEEE Transactions on Intelligent Transportation Systems}, 20(8):2986--2998, 2019.

\bibitem{wei2022game}
Chao Wei, Yuanhao He, Hanqing Tian, and Yanzhi Lv.
\newblock Game theoretic merging behavior control for autonomous vehicle at highway on-ramp.
\newblock {\em IEEE Transactions on Intelligent Transportation Systems}, 23(11):21127--21136, 2022.

\bibitem{huainingwu2022stochasticstability}
Huai-Ning Wu and Xiu-Mei Zhang.
\newblock Stochastic stability analysis and synthesis of a class of human-in-the-loop control systems.
\newblock {\em IEEE Transactions on Systems, Man, and Cybernetics: Systems}, 52(2):822--832, 2022.

\bibitem{XING2021humancenteredcollaborative}
Yang Xing, Chen Lv, Dongpu Cao, and Peng Hang.
\newblock Toward human-vehicle collaboration: Review and perspectives on human-centered collaborative automated driving.
\newblock {\em Transportation Research Part C: Emerging Technologies}, 128:103199, 2021.

\bibitem{YANG2018228}
Da~Yang, Shiyu Zheng, Cheng Wen, Peter~J. Jin, and Bin Ran.
\newblock A dynamic lane-changing trajectory planning model for automated vehicles.
\newblock {\em Transportation Research Part C: Emerging Technologies}, 95:228--247, 2018.

\bibitem{yang2018stackelberg}
Kaiming Yang, Rencheng Zheng, Xuewu Ji, Yosuke Nishimura, and Kazuya Ando.
\newblock Application of stackelberg game theory for shared steering torque control in lane change maneuver.
\newblock In {\em 2018 IEEE Intelligent Vehicles Symposium (IV)}, pages 138--143, 2018.

\bibitem{Ye2020RLlanechanging}
Fei Ye, Xuxin Cheng, Pin Wang, Ching-Yao Chan, and Jiucai Zhang.
\newblock Automated lane change strategy using proximal policy optimization-based deep reinforcement learning.
\newblock In {\em 2020 IEEE Intelligent Vehicles Symposium (IV)}, pages 1746--1752, 2020.

\bibitem{YU2021103101}
Haiyang Yu, Rui Jiang, Zhengbing He, Zuduo Zheng, Li~Li, Runkun Liu, and Xiqun Chen.
\newblock Automated vehicle-involved traffic flow studies: A survey of assumptions, models, speculations, and perspectives.
\newblock {\em Transportation Research Part C: Emerging Technologies}, 127:103101, 2021.

\bibitem{YU2018gametheory}
Hongtao Yu, H.~Eric Tseng, and Reza Langari.
\newblock A human-like game theory-based controller for automatic lane changing.
\newblock {\em Transportation Research Part C: Emerging Technologies}, 88:140--158, 2018.

\bibitem{Zhang2024Stackelberg}
Qingyu Zhang, Reza Langari, H.~Eric Tseng, Shankar Mohan, Steven Szwabowski, and Dimitar Filev.
\newblock Stackelberg differential lane change game based on mpc and inverse mpc.
\newblock {\em IEEE Transactions on Intelligent Transportation Systems}, 25(8):8473--8485, 2024.

\bibitem{ZHANG2019207}
Tingru Zhang, Da~Tao, Xingda Qu, Xiaoyan Zhang, Rui Lin, and Wei Zhang.
\newblock The roles of initial trust and perceived risk in public’s acceptance of automated vehicles.
\newblock {\em Transportation Research Part C: Emerging Technologies}, 98:207--220, 2019.

\bibitem{zhao2024safe}
Chenguang Zhao and Huan Yu.
\newblock Safe merging of autonomous vehicles under temporal logic task.
\newblock In {\em 2024 IEEE 63rd Conference on Decision and Control (CDC)}, pages 5693--5698, 2024.

\end{thebibliography}
\end{document}